\documentclass[journal,10pt,onecolumn,draftclsnofoot,]{IEEEtran}
\usepackage{bm}
\usepackage{cite}
\usepackage{amsmath,amssymb,amsfonts}
\usepackage{algorithmic}
\usepackage{graphicx}
\usepackage{textcomp}
\usepackage{dsfont}

\usepackage{scalerel}
\usepackage{tikz}
\usetikzlibrary{svg.path}

\definecolor{subsectioncolor}{rgb}{0,0.541,0.855}

\definecolor{orcidlogocol}{HTML}{A6CE39}
\tikzset{
    orcidlogo/.pic={
        \fill[orcidlogocol] svg{M256,128c0,70.7-57.3,128-128,128C57.3,256,0,198.7,0,128C0,57.3,57.3,0,128,0C198.7,0,256,57.3,256,128z};
        \fill[white] svg{M86.3,186.2H70.9V79.1h15.4v48.4V186.2z}
        svg{M108.9,79.1h41.6c39.6,0,57,28.3,57,53.6c0,27.5-21.5,53.6-56.8,53.6h-41.8V79.1z M124.3,172.4h24.5c34.9,0,42.9-26.5,42.9-39.7c0-21.5-13.7-39.7-43.7-39.7h-23.7V172.4z}
        svg{M88.7,56.8c0,5.5-4.5,10.1-10.1,10.1c-5.6,0-10.1-4.6-10.1-10.1c0-5.6,4.5-10.1,10.1-10.1C84.2,46.7,88.7,51.3,88.7,56.8z};
    }
}

\newcommand\orcidicon[1]{\href{https://orcid.org/#1}{\mbox{\scalerel*{
                \begin{tikzpicture}[yscale=-1,transform shape]
                \pic{orcidlogo};
                \end{tikzpicture}
            }{|}}}}
            
\def\BibTeX{{\rm B\kern-.05em{\sc i\kern-.025em b}\kern-.08em
    T\kern-.1667em\lower.7ex\hbox{E}\kern-.125emX}}

\makeatletter
\def\fnum@figure{\textcolor{subsectioncolor}{\sf Fig.~\thefigure}}
\def\fnum@table{\textcolor{subsectioncolor}{\sf TABLE~\thetable}}
\makeatother

\usepackage[hidelinks]{hyperref}

\begin{document}
\title{Two-Stage Electricity Markets with Renewable Energy Integration: Market Mechanisms and Equilibrium Analysis}
\author{Nathan Dahlin$^{\textsuperscript{\orcidicon{0000-0002-7006-251X}}}$ 
and Rahul Jain$^{\textsuperscript{\orcidicon{0000-0003-3786-8682}}}$ 
\thanks{This work was supported by NSF Awards ECCS-1611574 and ECCS-1810447.}
\thanks{Nathan Dahlin and Rahul Jain are with the Department of Electrical and Computer Engineering, University of Southern California, 3740 McClintock Ave, Los Angeles, CA, 90089, USA (e-mail: dahlin@usc.edu, rahul.jain@usc.edu, phone: (213) 631 6101).}
}

\newtheorem{theorem}{Theorem}
\newtheorem{lemma}{Lemma}
\newtheorem{proposition}{Proposition}
\newtheorem{corollary}{Corollary}[theorem]
\newtheorem{definition}{Definition}
\newtheorem{example}{Example}
\newtheorem{remark}{Remark}
\newtheorem{assumption}{Assumption}
\newtheorem{proof}{Proof}
\newtheorem{claim}{Claim}

\maketitle

\begin{abstract}
    We consider a two-stage market mechanism for trading electricity including renewable generation as an alter- native to the widely used multi-settlement market structure. The two-stage market structure allows for recourse decisions by the market operator, which are not possible in today's markets. We allow for different conventional generation cost curves in the forward and the real-time stages. We have considered costs of demand response programs and black outs, and adopt a DC power flow model to account for network constraints. Our first result is to show existence (by construction) of a sequential competitive equilibrium (SCEq) in such a two-stage market. We argue social welfare properties of such an SCEq, and then design a market mechanism that achieves social welfare maximization when the market participants are non-strategic. We also show that under either a \textit{congestion-free} or a \textit{monopoly-free} condition, an efficient Nash equilibrium exists.
\end{abstract}


\section{Introduction}
\label{sec:introduction}
Electricity markets in the US (and most of the world) are operated as multi-settlement markets. A large fraction of electricity is traded in  organized markets \cite{DBLP:conf/hicss/MoyeM18}. These are run by an independent system operator (ISO) as multiple forward markets, and a real-time or spot market. The forward markets operate at various time-scales: day-ahead, hour-ahead, etc. while the real-time market is opened five minutes prior. These markets are operated independently of each other, i.e.,  without \textit{recourse} by the market operator though the decision-making of the participants is obviously coupled. 

Until a few years ago, the primary uncertainty while trading in forward markets was in demand forecasts. As we approach real-time, this uncertainty reduces to 5\%, or lower \cite{forecasterror}. With the high penetration of renewables we are seeing today, there is increasing uncertainty in generation as well. For example, in California, renewables account for an estimated nearly 34\% of the total retail energy sales \cite{CECreportJan2019}. In fact, the state has mandated 100\% of power to come from renewables by 2045 \cite{cal100percent}. In other countries, most notably Germany, it has been reported that at times 100\% of power came from renewables \cite{germany100percent}. As can be imagined, this has introduced an order of magnitude greater uncertainty in net demand (demand minus renewable generation) than the current market structures have been able to handle. There is thus a need for new stochastic electricity market designs that can handle such uncertainty.

Given that economic dispatch is a multi-stage process, it only makes sense to couple the markets across various timescales (in various forward and the real-time markets) and allow for recourse decisions. This can enable achievement of greater efficiencies through the marketplace than are possible with the current multi-settlement market structure albeit with stochastic objectives of optimality. While it is an obvious point, it is worth repeating that if appropriate market architectures are not used for electricity trading, this can lead to inefficiencies and even affect reliability. The California electricity market debacle of early 2000s serves as an important lesson \cite{borenstein2002trouble}. 

In this paper, we study how multi-stage markets for electricity trading may work. We include multiple generators and load serving entities (LSEs) and consider two-stages: a stage 1 (forward stage) where only a forecast of renewable generation is available, and a stage 2 (the real-time stage), when exact realization of renewable generation is available. A dispatch decision is taken in stage 1. The economic dispatch problem is formulated as a two-stage stochastic program with recourse. The recourse decision in stage 2 is used to achieve power balance. Uncertainty in demand is ignored though it only adds to uncertainty in net demand which does not change things much.  
The generators own both dispatchable (i.e.,  controllable) primary and ancillary plants, while each LSE owns a nondispatchable or stochastic, renewable source, as is becoming more common \cite{utilityrenewables}.
In addition, it is assumed that the LSEs run a demand response (DR) program that gives each a lever to curtail demand to some extent, and at a cost, in scenarios where generation falls short of the overall demand  \cite{paterakis2017overview}. As an emergency recourse, the LSEs may schedule rolling blackouts at high (societal) cost, which thus must be avoided to the extent possible. We regard all participants to be non-strategic and acting as \textit{price-takers}. Building upon our model in \cite{dahlin2019twostage}, we assume a DC power flow model for the network which includes transmission line capacity constraints.

Our main result is the proof of existence of a sequential competitive equilibrium in this two-stage market with recourse, i.e.,  we show existence of first and second stage prices such that the first and second stage generation and consumption decisions achieve market clearance in stage 2, and power balance is achieved in real-time. We also establish analogues of the first and second fundamental theorems of welfare economics in the multi-stage setting. Together these theorems state that the sequential  competitive equilibrium market allocation achieves market efficiency, i.e.,  social welfare maximization, and market efficiency can be supported by such an equilibrium. 
We then outline a two-stage market mechanism that can be used for multi-stage economic dispatch. Finally, as an extension on our results in \cite{dahlin2019twostage}, we show that under either a \textit{congestion-free} or a \textit{monopoly-free} condition, an efficient Nash equilibrium exists when the LSEs bid marginal optimal value. On the other hand, if the LSEs bid marginal demand response costs, we show through a counterexample that in that case an efficient Nash equilibrium may not exist.

\textbf{Related work.} \emph{Dynamic competitive equilibria:} The closest work to this paper is the dynamic competitive equilibrium framework in \cite{wannegkowshameysha11b} and \cite{wang2011control} which considers a much more general setting, but focuses on showing existence of such a competitive equilibrium, while we are motivated by design of a multi-stage market mechanism for economic dispatch that can be used to achieve it. 
Other related work is \cite{DBLP:conf/cdc/GuptaJPV15}, which does sequential equilibrium analysis but in the current multi-settlement market, and \cite{tang2016dynamic}, which examines the need for dispatching storage, again in current multi-settlement markets, concluding that storage decisions can be left to other entities and the system operator need not dispatch it without any loss in market efficiency. Neither of these papers address the issue of alternative market designs for multi-stage economic dispatch with uncertain generation and recourse action.

\emph{Nash equilibria in electricity markets:} Since the 1950's, the most widely used solution concept in applications of game theory to economics is that of \emph{Nash Equilibrium} \cite{mas1995microeconomic}. Still, the concept has been relatively understudied in the context of electricity markets. In the context of the storage enabled setting of \cite{tang2016dynamic}, sufficient conditions for the existence of an efficient, multi stage Nash equilibrium similar to those presented here are given. In a one-shot Cournot competition setting, \cite{bose2014role} studies the effect of the market maker objective function on the equilibrium existence, and associated network flows. Both of these works focus on the strategic incentives of generators only, while we analyze strategic behavior of both generators and LSEs. 
\section{Electricity Market Model}\label{sec:model}

We consider a connected power network consisting of bus (node) set $\mathcal{N}$, with $|\mathcal{N}| = N$. We index individual nodes with $i$. Located at each node $i$ is a set of generators, $\mathcal{G}_i$, and a set of LSEs $\mathcal{L}_i$ (possibly with $|\mathcal{G}_i|\neq |\mathcal{L}_i|$). The LSEs can be thought of as electric utilities serving a population of end users. We will index individual generators or LSEs at a given node $i$ by $k$, and reference the $k$th generator or LSE at node $i$ as $(i,k)$. Apart from the generators and LSEs, an \emph{independent system operator} (ISO) operates the grid, and also plays the role of the social planner.

We study a two-stage setting, with generation dispatched in the first stage (also referred to as the day ahead or DA stage), and subsequently adjusted in the second stage (real time or RT) in order to fulfill demand at each node. In the following, variables with first subscript $\ell=1,2$ correspond to the $\ell$th market stage. Often when referring to decision variables corresponding to an individual generator or LSE, we will simply use the subscript $k$ and suppress the nodal index $i$. 

For each $k\in\mathcal{L}_i$ and $i\in\mathcal{N}$, we denote the aggregate consumer demand to be met by LSE $(i,k)$ as $D_{k}\geq0$. In our setting, we assume that this aggregate demand is \emph{inelastic}, meaning it is not sensitive to changes in first or second stage electricity prices. 

LSE $(i,k)$'s stochastic, renewable generation $W_k$ is assumed to take on finitely many values in the interval $[0,\overline{w}_k]$, with $\overline{w}_k>0$ for all $k\in\mathcal{L}_i$ and $i\in\mathcal{N}$. Together, the outputs of each LSE's renewable output form the random vector $W$. This vector has associated probability mass function $p=(p_0,\dots,p_{\overline{w}})$, assumed to be common knowledge amongst all market participants, with the probability of scenario (realization) $W=w\in\mathbb{R}^{|\mathcal{L}|}_+$ given by $p_w$, where $\mathcal{L}$ is the set of all LSEs in the network. In all scenarios and at all buses, renewable generation incurs zero marginal cost. Given realized renewable generation $W=w$, we denote the quantity of renewable energy produced by LSE $(i,k)$ as $w_k$ for all $k\in\mathcal{L}_i$ and $i\in\mathcal{N}$. 

At each node $i\in\mathcal{N}$, LSE $i$ may purchase energy in both the DA and RT markets, at nodal prices $P_{1,i}$ and $P_{2,i}(w)$ (given $W=w$), respectively, where $P_{2,i}\,:\,\mathbb{R}^{|\mathcal{L}|}_+\to\mathbb{R}$. Note that in our finite scenario setting, each $P_{2,i}(\cdot)$ may be considered as a finite length vector. We denote the quantities purchased by LSE $(i,k)$ in the first and second stages, given $W=w$ and $i\in\mathcal{N}$ as $y^L_{1,k}$ and $y^L_{2,k}(w)$ for all $k\in\mathcal{L}_i$. 

Each LSE $i$'s constituent consumers are assumed to participate in \emph{demand response} programs, where they are compensated for reducing their consumption at the LSE's request. In more detail, having secured first stage energy quantity $y^L_{1,k}$ and observed $W=w$, LSE $i$ may request its consumer population to reduce their aggregate energy by amount $x^L_{2,k}(w)$, incurring consumer compensation cost $c_{\text{dr},k}(x^L_{2,k}(w))$. 

While the demand response programs are intended to provide LSEs with a cushion against second stage energy prices, in cases of extreme underproduction in renewables, LSEs may effectively schedule a shortfall in energy offered to its consumers by selecting energy and demand reduction quantities in the RT market which sum to less than the \emph{residual demand} $D_k-y^L_{1,k}$. Such an action incurs blackout cost $c_{\text{bo},k}(z^L_{2,k}(w))$, with 
$$z^L_{2,k}(w) = D_{k}-w_k-y^L_{1,k} - y^L_{2,k}(w)-x^L_{2,k}(w).$$

It is assumed that each LSE $i$ is \emph{price taking}, so that its consumption decisions $y^L_{1,k}$ and $y^L_{2,k}(w)$ cannot affect energy prices in either market stage. Thus, given $P_{1,k}$ and $P_{2,k}(\cdot)$, renewable generation vector $W=w$, as well as service decisions $(y^L_{1,k},y^L_{2,k}(w),x^L_{2,k}(w),z^L_{2,k}(w))$ the utility enjoyed by LSE $(i,k)$ is 
\begin{equation}
\label{LSEprofit}
\begin{split}
&\pi^L_{i,k}(y^L_{1,k},y^L_{2,k}(w),x^L_{2,k}(w),z^L_{2,k}(w)):=-P_{1,i}y^L_{1,k} - P_{2,i}(w)y^L_{2,k}(w)- c_{\text{dr},k}(x^L_{2,k}(w)) - c_{\text{bo},k}(z^L_{2,k}(w)).
\end{split}
\end{equation}

Each LSE $(i,k)$ seeks to maximize (\ref{LSEprofit}) with its first and second stage decisions. Specifically, given $y^L_{1,k}$, renewable generation $W=w$, and second stage nodal price schedule $P_{2,i}$, the LSE $(i,k)$'s second stage optimization problem is 
\begin{align}
(\text{LSE2}_{i,k})\,\max_{\substack{y^L_{2,k}(w),x^L_{2,k}(w)\\z^L_{2,k}(w)}}&\,-P_{2,k}(w)y^L_{2,k}(w) - c_{\text{dr},k}(x^L_{2,k}(w))- c_{\text{bo},k}(z^L_{2,k}(w))\\
\label{LSEdemand_const}\text{s.t.}&\quad y^L_{1,k} + y^L_{2,k}(w) + x^L_{2,k}(w) + z^L_{2,k}(w) \geq D_k - w_k\\
&\quad y^L_{2,k}(w)\geq0,\,x^L_{2,k}(w)\geq 0,\,z^L_{2,k}(w)\geq 0
\end{align}

Note that constraint (\ref{LSEdemand_const}) is an inequality in order to maintain feasibility when $D_k-y^L_{1,k}-w_k<0$, i.e.,  when renewable generation exceeds residual demand $D_k-y^L_{1,k}$. 

Denote as $\pi^L_{2,k}(y^L_{1,k};w,P_{2,k})$ the maximum utility achievable in ($\text{LSE2}_{i,k}$), given LSE $(i,k)$'s first stage decision, realized renewable generation and prices. Then, given prices $P_{1,i}$ and $P_{2,i}$, LSE $(i,k)$'s first stage optimization problem is to maximize its summed first stage utility $-P_{1,i}y^L_{1,k}$ and \emph{expected} maximum second stage utility with respect to uncertainty in renewable generation:
\begin{align}
(\text{LSE1}_{i,k})&\quad \max_{y^L_{1,k}}\quad -P_{1,i}y^L_{1,k} + \mathbb{E}[\pi^L_{2,k}(y^L_{1,k};W,P_{2,i})]\\
\text{s.t.}&\quad y^L_{1,k}\geq 0.
\end{align}

Each generator is equipped with two different sources of power generation. First, generator $(i,k)$ owns a primary, dispatchable, nonrenewable power station, which can be scheduled to produce quantity $y^G_{1,k}\in\mathbb{R}_+$ at cost $c_{1,k}(y^G_{1,k})$. This primary plant is assumed to be inflexible for the purposes of our market, i.e.,  once scheduled in the first stage its generation level must remain fixed in the second stage.  

Each generator $(i,k)$ also owns and operates a secondary or ancillary station, e.g., a gas turbine, which can be dispatched quickly in the second stage market. Specifically, having scheduled its primary plant to produce energy amount $y^G_{1,k}$ and observed renewable generation $W=w$, the generator may schedule secondary generation quantity $y^G_{2,k}(w)\in\mathbb{R}_+$, incurring cost $c_{2,k}(y^G_{2,k}(w))$. 

Throughout, we assume that dispatchable or renewable energy produced in excess of consumer demand may be disposed of at zero cost, or placed in a separate spot market not considered here. 

Generation $(i,k)$ is compensated for first stage generation $y^G_{1,k}$ at rate $P_{1,i}$ and, given $W=w$, compensated for second stage generation $y^G_{2,k}(w)$ at rate $P_{2,k}(w)$. As with the LSEs, we assume that each generator $(i,k)$ is price taking, so that given prices $P_{1,i}$, $P_{2,i}$ and $W=w$, along with dispatch decisions $y^G_{1,k}$ and $y^G_{2,k}(w)$, generator $(i,k)$ enjoys profit
\begin{equation}\label{genprofit}
\begin{split}
&\pi^G_{i,k}(y^G_{1,k},y^G_{2,k}(w)) := P_{1,i}y^G_{1,k} - c_{1,k}(y^G_{1,k}) + P_{2,k}(w)y^G_{2,k}(w) - c_{2,k}(y^G_{2,k}(w)).
\end{split}
\end{equation}

Given $P_{2,k}$ and renewable generation scenario $W=w$, generator $(i,k)$ maximizes its second stage profit by solving 
\begin{align}
(\text{GEN2}_{i,k})\quad \max_{y^G_{2,k}(w)}&\quad P_{2,i}(w)y^G_{2,k}(w) - c_{2,k}(y^G_{2,k}(w))\\
\text{s.t.}&\quad y^G_{2,k}(w)\geq0.
\end{align}
Let $\pi^G_{2,k}(w,P_{2,i})$ denote generator $(i,k)$'s maximum achievable second stage profit, given $W=w$ and $P_{2,i}$. Then, in the first stage, generator $(i,k)$ solves
\begin{align}
(\text{GEN1}_{i,k})\quad \max_{y^G_{1,k}}&\quad P_{1,i}y^G_{1,k} - c_{1,k}(y^G_{1,k}) + \mathbb{E}[\pi^G_{2,k}(W,P_{2,i})]\\
\text{s.t.}&\quad y^G_{1,k}\geq 0.
\end{align}

Note here that the expected second stage profit $\mathbb{E}[\pi^G_{2,k}(W,P_{2,i})]$ is a constant when optimizing over $y^G_{1,k}$, reflecting the fact that from the view of each generator $(i,k)$, the two market stage are completely decoupled. We separate the two generator optimization problems to emphasize that the generator observes $W=w$ before it decides $y^G_{2,k}(w)$. 

Finally, the ISO is responsible for enforcing the safe operation of the power grid, which we describe using the DC power flow model \cite{stott2009dc}. The model characterizes the network lines with susceptance matrix $B$, where $B_{ij} = B_{ji}$ gives the susceptance of the line connecting nodes $i$ and $j$. Denote the voltage phase angle at node $i$ in the DA stage as $\theta_{1,i}$, and in the RT stage, given scenario $w$, as $\theta_{2,i}(w)$. Then, the active power flows from node $i$ to node $j$ in each stage are 
$$f_{1,ij} = B_{ij}(\theta_{1,i} - \theta_{1,j}),\quad f_{2,ij}(w) = B_{ij}(\theta_{2,i}(w) - \theta_{2,j}(w)),$$
and the power balance equations at node $i$ in each stage are 
\begin{equation}
\label{powerbalstages}
\begin{split}
\sum_{k\in\mathcal{G}_i}y^G_{1,k} - \sum_{k\in\mathcal{L}_i}y^L_{1,k} &= \sum_jf_{1,ij},\quad\sum_{k\in\mathcal{G}_i}y^G_{2,k}(w) - \sum_{k\in\mathcal{L}_i}y^L_{2,k}(w) = \sum_jf_{2,ij}(w) - \sum_jf_{ij,1}
\end{split}
\end{equation}

Note that summing both equalities in (\ref{powerbalstages}) gives that 
\begin{equation}
\begin{split}
\sum_i\sum_{k\in\mathcal{G}_i}y^G_{1,k} = \sum_i\sum_{k\in\mathcal{L}_i}y^L_{1,k},\quad\sum_i\sum_{k\in\mathcal{G}_i}y^G_{2,k}(w) = \sum_i\sum_{k\in\mathcal{L}_i}y^L_{2,k}(w), \quad\forall\,w.\end{split}
\end{equation}

Letting $f^{\max}_{ij} = f^{\max}_{ji}\geq 0$ denote the flow limit of line $(i,j)$, the ISO also ensures that power flows do not exceed line capacities in either market stage: 
$$f_{1,ij}\leq f^{\max}_{ij}\quad\text{and}\quad f_{2,ij}(w)\leq f^{\max}_{ij}, \quad\forall\,i,j,w$$

In the following sections we will introduce a sequential competitive equilibrium definition for our setting. In order to assess the welfare properties of the allocations included in such equilibria, we specify here a two-stage social planner's problem (SPP) which corresponds to our two settlement market model. The social planner seeks to maximize the aggregate welfare of all market participants. 

Given renewable generation scenario $W=w$, the aggregate welfare is defined as the sum of LSE utilities and generator profits given in (\ref{LSEprofit}) and (\ref{genprofit}): 
\begin{equation}
\begin{split}
\pi^{\text{SPP}}(w)&:= \sum_i\sum_{k\in\mathcal{G}_i}(\pi^G_{i,k}(\hat{y}^G_{1,k},\hat{y}^G_{2,k}(w))+ \sum_{k\in\mathcal{L}_i}\pi^L_{i,k}(\hat{y}^L_{1,k},\hat{y}^L_{2,k}(w),\hat{x}^L_{2,k}(w),\hat{z}^L_{2,k}(w))\\
&= -\sum_i\sum_{k\in\mathcal{G}_i}\left(c_{1,k}(\hat{y}^G_{1,k} + \hat{y}^G_{2,k}(w)\right)- \sum_i\sum_{k\in\mathcal{L}_i}\left(c_{\text{dr},k}(x^L_{2,k}(w))+c_{\text{bo},k}(z^L_{2,k}(w)),\right)
\end{split}
\end{equation}
where $(\hat{y}^G_{1,k},\hat{y}^L_{1,k},\hat{y}^G_{2,k}(w),\hat{y}^L_{2,k}(w),\hat{x}^L_{2,k}(w),\hat{z}^L_{2,k}(w))$ for all $k\in\mathcal{G}_i$ and $\mathcal{L}_i$ and all $i\in\mathcal{N}$ are the planner's first and second stage decisions, the second stage decisions made with knowledge of realized scenario $W=w$. 

Let $\hat{y}^G_1\in\mathbb{R}_+$ denote the vector collecting the planner's first stage generation dispatch decisions for all generators. Similarly defining $\hat{y}^L_1$ and $\hat{\theta}_1$, given $\hat{y}_1 := (\hat{y}^G_1,\hat{y}^L_1,\hat{\theta}_1)$ and $W=w$, the social planner's RT optimization problem is 
\begin{align}
\text{(SPP2)}\,\max_{\substack{\hat{y}^G_{2}(w),\hat{y}^L_2(w)\\\hat{x}^L_2(w),\hat{z}^L_2(w),\hat{\theta}_2(w)}}&\,-\sum_i\sum_{k\in\mathcal{G}_i}c_{2,k}(\hat{y}^G_{2,k}(w))-\sum_i\sum_{k\in\mathcal{G}_i}\left(c_{\text{dr},k}(\hat{x}^L_{2,k}(w)) + c_{\text{bo},k}(\hat{z}^L_{2,k}(w))\right)p_w\\
\text{s.t.}&\quad \sum_{k\in\mathcal{G}_i}\hat{y}^G_{2,k}(w) - \sum_{k\in\mathcal{L}_i}\hat{y}^L_{2,k}(w) =\sum_jB_{ij}(\hat{\theta}_{2,i}(w) - \hat{\theta}_{2,j}(w) - \hat{\theta}_{1,i} +\hat{\theta}_{1,j})\\
&\quad B_{ij}(\hat{\theta}_{2,i}(w) - \hat{\theta}_{2,j}(w))\leq f^{\max}_{ij}, \quad\forall\,i,j\\
\nonumber&\quad \hat{y}^L_{1,k} + \hat{y}^L_{2,k}(w)\\
&\hspace{0.36in}+ \hat{x}^L_{2,k}(w) + \hat{z}^L_{2,k}(w) \geq D_k - w_k, \quad\forall\,i,k\in\mathcal{L}_i\\
&\quad \hat{y}^L_{2,k}(w)\geq 0,\,\hat{x}^L_{2,k}(w)\geq0,\,\hat{z}^L_{2,k}(w)\geq 0, \,\,\forall\,i,k\in\mathcal{L}_i\\
&\quad \hat{y}^G_{2,k}(w)\geq 0\quad \forall\,i,k\in\mathcal{G}_i.
\end{align}
Define $\pi^{\text{SPP}}_2(\hat{y}^G_1,\hat{y}^L_1,\hat{\theta}_1;w)$ as the maximum aggregate welfare achievable in the second stage, given the planner's first stage decisions and $W=w$. Then, the planner's first stage problem is
\begin{align}
(\text{SPP1})\max_{\hat{y}^G_1,\hat{y}^L_1,\hat{\theta}_1}&\quad- \sum_i\sum_{k\in\mathcal{G}_i}c_{1,k}(\hat{y}^G_{1,k}) + \mathbb{E}[\pi^{\text{SPP}}_2(\hat{y}^G_1,\hat{y}^L_1;W)]\\
\text{s.t.}&\quad \sum_{k\in\mathcal{G}_i}\hat{y}^G_{1,k} - \sum_{k\in\mathcal{L}_i}\hat{y}^L_{1,k} = \sum_jB_{ij}(\hat{\theta}_i-\hat{\theta}_j), \quad\forall\,i\\
&\quad B_{ij}(\hat{\theta}_i-\hat{\theta}_j)\leq f^{\max}_{ij}\quad \forall\,i,j\\
&\quad \hat{y}^G_{1,i}\geq 0,\,\hat{y}^L_{1,i}\geq0, \quad\forall\,i.
\end{align}
Let $\hat{y}_2(w):=(\hat{y}^G_2(w),\hat{y}^L_2(w))$ for all $w$, and similarly for $\hat{x}_2(w)$. We refer to optimal solutions $\hat{y}^*_1$ and $(\hat{y}^*_2(\cdot),\hat{x}^*_2(\cdot),\hat{z}^*_2(\cdot))$ to (SPP1) and (SPP2) as \emph{efficient sequential allocations}. The following assumptions are made throughout the following sections.
\begin{assumption}
$c_{1,k},c_{2,k},c_{\text{dr},k}$ and $c_{\text{bo},k}$ are strictly convex, increasing, differentiable and nonnegative over $\mathbb{R}_+$ for all generators and LSEs. 
\end{assumption}

First, we argue that problems (SPP1) and (SPP2) can be combined into a single-stage optimization problem. 
\begin{lemma}
The two-stage problem (SPP1)-(SPP2) is equivalent to the following primal single stage problem:
\begin{align}
\nonumber\text{(SPP-P)}\max_{\substack{\hat{y}^G_1,\hat{y}^G_2,\hat{y}^L_1,\hat{y}^L_2\\\hat{x}^L_2,\hat{z}^L_2,\hat{\theta}_1,\hat{\theta}_2}}&\quad-\sum_i\sum_{k\in\mathcal{G}_i}c_{1,k}(\hat{y}^G_{1,k})+\sum_{w}c_{2,k}(\hat{y}^G_{2,k}(w)p_w\\
&\quad\hspace{1in}-\sum_i\sum_{k\in\mathcal{L}_i}\sum_w\left(c_{\text{dr},k}(\hat{x}^L_{2,k}(w)) + c_{\text{bo},k}(\hat{z}^L_{2,k}(w))\right)p_w\\
\label{SPP_powerbal1}\text{s.t.}&\quad \sum_{k\in\mathcal{G}_i}\hat{y}^G_{1,k} - \sum_{k\in\mathcal{L}_i}\hat{y}^L_{1,k} = \sum_jB_{ij}(\hat{\theta}_{1,i}-\hat{\theta}_{1,j}), \,\,\forall\,i\\
\label{SPP_powerbal2}&\quad \sum_{k\in\mathcal{G}_i}\hat{y}^G_{2,k}(w) - \sum_{k\in\mathcal{L}_i}\hat{y}^L_{2,k}(w) =\sum_jB_{ij}(\hat{\theta}_{2,i}(w)-\hat{\theta}_{2,j}(w)-\hat{\theta}_{1,i}+\hat{\theta}_{1,j}), \,\,\forall\,i,w\\
&\quad B_{ij}(\hat{\theta}_{1,i} - \hat{\theta}_{1,j})\leq f^{\max}_{ij}, \quad\forall\,i,j\\
&\quad B_{ij}(\hat{\theta}_{2,i}(w) - \hat{\theta}_{2,j}(w))\leq f^{\max}_{ij}, \quad\forall\,i,j\\
&\quad \hat{y}^L_{1,k} + \hat{y}^L_{2,k}(w)+\hat{x}^L_{2,k}(w) + \hat{z}^L_{2,k}(w)\geq D_k - w_k, \,\,\forall\,i,k\in\mathcal{L}_i,w\\
&\quad \hat{y}^L_{1,k}\geq 0\quad\forall\,i,k\in\mathcal{L}_i\\
&\quad \hat{y}^L_{2,k}(w)\geq0,\hat{x}^L_{2,k}(w)\geq0,\hat{z}^L_{2,k}(w)\geq0,\,\,\forall\,i,k\in\mathcal{L}_i,w\\
\label{SPP_const_end}&\quad \hat{y}^G_{1,k}\geq0,\hat{y}^L_{2,k}(w)\geq 0, \quad\forall\,i,k\in\mathcal{G}_i,w
\end{align}
\end{lemma}
\begin{proof}
The proof of Lemma 1 follows the same method as the proof of Lemma 1 in \cite{dahlin2019twostage}.
\end{proof}
By ``equivalent'' in Lemma 1, we mean that (SPP-P) and (SPP1) have the same optimal objective value. Moreover, if $(\hat{y}^{G*}_1,\hat{y}^{L*}_1,\hat{y}^{G*}_2(\cdot),\hat{y}^{L*}_2(\cdot),\hat{x}^{L*}_2(\cdot),\hat{z}^{L*}_2(\cdot),\hat{\theta}^*_1,\hat{\theta}^*_2(\cdot))$ is optimal for (SPP-P), then $(\hat{y}^{G*}_2(\cdot),\hat{y}^{L*}_2(\cdot),\hat{x}^{L*}_2(\cdot),\hat{\theta}^*_2(\cdot))$ is optimal for (SPP2). Conversely, if $(\hat{y}^{G*}_1,\hat{y}^{L*}_1,\hat{\theta}^*_1)$ is optimal for (SPP1) and $(\hat{y}^{G*}_2(\cdot),\hat{y}^{L*}_2(\cdot),\hat{x}^{L*}_2(\cdot),\hat{z}^{L*}_2(\cdot),\hat{\theta}^*_2(\cdot))$ is optimal for (SPP2), given the optimal solution to (SPP1), then $(\hat{y}^{G*}_1,\hat{y}^{L*}_1,\hat{y}^{G*}_2(\cdot),\hat{y}^{L*}_2(\cdot),\hat{x}^{L*}_2(\cdot),\hat{z}^{L*}_2(\cdot),\hat{\theta}^*_1,\hat{\theta}^*_2(\cdot))$ is optimal for (SPP-P). Similar results hold for the generator and LSE two-stage problems, giving 
\begin{align}
(\text{GEN-P}_{i,k})\max_{y^G_{1,k},y^G_{2,k}}&\quad P_{1,i}y^G_{1,k} - c_{1,k}(y^G_{1,k})+ \sum_w\left(P_{2,i}(w)y^G_{2,k}(w) - c_{2,k}(y^G_{2,k}(w))\right)p_w\\
\label{gen_nonneg}\text{s.t.}&\quad y^G_{1,k}\geq0,\,y^G_{2,k}(w)\geq0\quad \forall\,w,
\end{align}
and 
\begin{align}
(\text{LSE-P}_{i,k})\max_{\substack{y^L_{1,k},y^L_{2,k}\\x^L_{2,k},z^L_{2,k}}}&\quad -P_{1,i}y^L_{1,k}-\sum_wP_{2,i}(w)y^L_{2,k}(w)p_w-\sum_w\left(c_{\text{dr},k}(x^L_{2,k}(w))+c_{\text{bo},k}(z^L_{2,k}(w))\right)p_w\\
\label{LSEP_const1}\text{s.t.}&\quad y^L_{1,k} + y^L_{2,k}(w)+x^L_{2,k}(w) + z^L_{2,k}(w) \geq D_k-w_k, \quad\forall\,w\\
\label{LSEP_const2}&\quad y^L_{1,k}\geq0,y^L_{2,k}(w)\geq0,x^L_{2,k}(w)\geq0,z^L_{2,k}(w)\geq0, \,\forall\,w.
\end{align}

\section{Sequential Competitive Equilibrium and Efficient Allocations}\label{sec:equilibrium}

In single-stage markets involving a single good, a \emph{competitive equilibrium} is given by a price $P$ and quantity $x$ such that, given $P$, producers find it optimal to produce, and consumers find it optimal to purchase quantity $x$ of the good \cite{mas1995microeconomic}. In such a situation, it is said that the \emph{market clears}, i.e.,  demand equals supply. 

Here, in order to assess the outcome of the two-stage market described in the previous section, we define a two-stage, sequential version of competitive equilibrium, similar to that found in \cite{dahlin2019twostage}. 
\begin{definition}
A \emph{sequential competitive equilibrium} (SCEq) is a tuple $(\overline{y}^{G*}_1,\overline{y}^{L*}_1,\overline{y}^{G*}_2(\cdot),\overline{y}^{L*}_2(\cdot),\overline{x}^{L*}_2(\cdot),P^*_1,P^*_2(\cdot))$ such that, for all $(i,k)$, given $P^*_{1,i}$ and $P^*_{2,i}(\cdot)$, $\overline{y}^{G*}_{1,k}$ is optimal for $(\text{GEN1}_{i,k})$, $\overline{y}^{L*}_{i,k}$ is optimal for $(\text{LSE1}_{i,k})$, and, given $W=w$ and $P^*_2(\cdot)$, $\overline{y}^{G*}_{2,k}(w)$ is optimal for $(\text{GEN2}_{i,k})$ and $(\overline{y}^{L*}_{2,k}(w),\overline{x}^{L*}_{2,k}(w))$ is optimal for $(\text{LSE2}_{i,k})$, and the markets clear in both stages in all scenarios:
\begin{equation}
\begin{split}
    \sum_i\sum_{k\in\mathcal{G}_i}\overline{y}^{G*}_{1,k} = \sum_i\sum_{k\in\mathcal{L}_i}\overline{y}^{L*}_{1,k},\quad \sum_i\sum_{k\in\mathcal{G}_i}\overline{y}^{G*}_{2,k}(w) = \sum_i\sum_{k\in\mathcal{L}_i}\overline{y}^{L*}_{2,k}(w), \quad\forall\,w.
    \end{split}
\end{equation}
\end{definition}
Note that in the SCEq definition, $P^*_2(\cdot)$ and $\overline{y}^{G*}_{2,k}(\cdot)$, $\overline{y}^{L*}_{2,k}(\cdot)$ and $\overline{x}^{L*}_{2,k}(\cdot)$ for each $k$ are functions. Also in problems involving shortfall decisions, a solution consists only of ancillary generation and demand response decisions, as they are enough to determine shortfall decisions, e.g., $(\hat{y}^{G*}_2(\cdot),\hat{y}^{L*}_2(\cdot),\hat{x}^{L*}_2(\cdot))$ solves (SPP2). 

We now study the existence of an SCEq in the two-stage market. Let $\hat{\lambda}^*_1 = (\lambda^*_{1,1},\cdots,\lambda^*_{1,N})^{\top}$ and $\hat{\lambda}^*_2(\cdot) = (\hat{\lambda}^*_{2,1}(\cdot),\dots,\hat{\lambda}^*_{2,N}(\cdot))^{\top}$ denote dual optimal variables associated with constraints (\ref{SPP_powerbal1}) and (\ref{SPP_powerbal2}) in (SPP-P). 

\begin{theorem}\label{compeq_existence}
Under Assumption 1, a sequential competitive equilibrium exists, and is given by $$(\hat{y}^{G*}_1,\hat{y}^{L*}_1,\hat{y}^{G*}_2(\cdot),\hat{y}^{L*}_2(\cdot),\hat{x}^{L*}_2(\cdot),\hat{\lambda}^*_1,\hat{\lambda}^*_2(\cdot)),$$ where $(\hat{y}^{G*}_1,\hat{y}^{L*}_1,\hat{y}^{L*}_2(\cdot),\hat{y}^{L*}_2(\cdot),\hat{x}^{L*}_2(\cdot))$ is the primal solution to (SPP-P), and $(\hat{\lambda}^*_1,\hat{\lambda}^*_2(\cdot))$ is an optimal dual solution to (SPP-P). 
\end{theorem}

\begin{proof}
In addition to feasibility constraints (\ref{SPP_powerbal1})-(\ref{SPP_const_end}), the optimal solution to (SPP-P), denoted as $$(\hat{y}^{G*}_1,\hat{y}^{L*}_1,\hat{y}^{G*}_2(\cdot),\hat{y}^{L*}_2(\cdot),\hat{x}^{L*}_2(\cdot),\hat{\theta}^*_1,\hat{\theta}^*_2(\cdot))$$ satisfies the following KKT conditions:
\begin{align}
\label{SPPKKT1}c'_{1,k}(\hat{y}^{G*}_{1,k}) - \hat{\lambda}^*_{1,i}&\geq 0, \,\,\forall\,i,k\in\mathcal{G}_i\\
\label{SPPKKT2}\hat{y}^{G*}_{1,k}\left(c'_{1,k}(\hat{y}^{G*}_{1,k}) - \hat{\lambda}^*_{1,i}\right)&=0, \,\,\forall\,i,k\in\mathcal{G}_i\\
\label{SPPKKT3}c'_{2,k}(\hat{y}^{G*}_{2,k}(w)) - \hat{\lambda}^*_{2,i}(w)&\geq 0, \,\,\forall\,i,k\in\mathcal{G}_i,w\\
\label{SPPKKT4}\hat{y}^{G*}_{2,k}(w)\left(c'_{2,k}(\hat{y}^{G*}_{2,k}(w)) - \hat{\lambda}^*_{2,i}(w)\right)&=0, \,\,\forall\,i,k\in\mathcal{G}_i\\
    \label{SPPKKT5}\hat{\lambda}^*_{1,i}-\sum_w\hat{\mu}^*_{2,k}(w)p_w&\geq 0, \,\,\forall\,i,k\in\mathcal{L}_i,w\\
    \label{SPPKKT6}\hat{y}^{L*}_{1,k}\left(\hat{\lambda}^*_{1,i}-\sum_w\hat{\mu}^*_{2,k}(w)p_w\right)&=0, \,\,\forall\,i,k\in\mathcal{L}_i,w\\
    \label{SPPKKT7}\hat{\lambda}^*_{2,i}(w) - \hat{\mu}^*_{2,k}(w)&\geq 0, \,\,\forall\,i,k\in\mathcal{L}_i,w\\
    \label{SPPKKT8}\hat{y}^{L*}_{2,k}(w)\left(\hat{\lambda}^*_{2,i}(w) - \hat{\mu}^*_{2,k}(w)\right)&=0, \,\,\forall\,i,k\in\mathcal{L}_i,w\\
\label{SPPKKT9}c'_{\text{dr},k}(\hat{x}^{L*}_{2,k}(w)) - \hat{\mu}^*_{2,k}(w)&\geq 0, \forall\,i,k\in\mathcal{L}_i,w\\
\label{SPPKKT10}\hat{x}^{L*}_{2,k}(w)\left(c'_{\text{dr},k}(\hat{x}^{L*}_{2,k}(w)) - \hat{\mu}^*_{2,k}(w)\right)&= 0, \forall\,i,k\in\mathcal{L}_i,w\\
\label{SPPKKT11}c'_{\text{bo},k}(\hat{z}^{L*}_{2,k}(w)) - \hat{\mu}^*_{2,k}(w)&\geq 0, \,\,\forall\,i,k\in\mathcal{L}_i,w\\
    \label{SPPKKT12}\hat{z}^{L*}_{2,k}(w)\left(c'_{\text{bo},k}(\hat{z}^{L*}_{2,k}(w)) - \hat{\mu}^*_{2,k}(w)\right)&= 0,\,\,\forall\,i,k\in\mathcal{L}_i,w\\
    \label{SPPKKT13}\hat{\mu}^*_{2,k}(w)(D_k-\hat{y}^{L*}_{1,k} - \hat{y}^{L*}_{2,k}(w)- \hat{x}^{L*}_{2,k}(w) - \hat{z}^{L*}_{2,k}(w)-w_k)&=0, \,\,\forall\,i,k\in\mathcal{L}_i,w\\
    \label{SPPKKT14}\sum_jB_{ij}\big(\hat{\lambda}^*_{1,i}-\hat{\lambda}^*_{1,j} - \sum_w\big(\hat{\lambda}^*_{2,i}(w)-\hat{\lambda}^*_{2,j}(w)\big)p_w+\hat{\gamma}^*_{1,ij}-\hat{\gamma}^*_{1,ji}\big)&=0, \,\,\forall\,i\\
    \label{SPPKKT15}\sum_jB_{ij}(\hat{\lambda}^*_{2,i}(w)-\hat{\lambda}^*_{2,j}(w)+\hat{\gamma}^*_{2,ij}(w)-\hat{\gamma}^*_{2,ji}(w))&=0, \,\,\forall\,i,w\\
    \label{SPPKKT16}\hat{\gamma}^*_{1,ij}\left(B_{ij}(\hat{\theta}^*_{1,i}-\hat{\theta}^*_{1,j})-f^{\max}_{ij}\right)&=0, \,\,\forall\,i,j\\
    \label{SPPKKT17}\hat{\gamma}^*_{2,ij}(w)\left(B_{ij}(\hat{\theta}^*_{2,i}(w)-\hat{\theta}^*_{2,j}(w))-f^{\max}_{ij}\right)&=0, \,\,\forall\,i,j,w\\
    \label{SPPKKT18}\hat{\mu}^*_{2,k}(w)&\geq 0 , \quad\forall\,i,k\in\mathcal{L}_i,w\\
\label{SPPKKT19}\hat{\gamma}^*_{1,ij}\geq 0,\hat{\gamma}^*_{2,ij}(w)&\geq 0\quad \forall\,i,j,w
\end{align}

Note that, due to Assumption 1, the optimal solution to (SPP-P) is unique when it exists. We show here that this solution also gives optimal solutions to $(\text{GEN-P}_{i,k})$ and $(\text{LSE-P}_{i,k})$ for all $(i,k)$. 

Aside from the nonnegativity constraints given in (\ref{gen_nonneg}), the optimal solution to $(\text{GEN-P}_{i,k})$ satisfies 
\begin{align}
\label{GENPKKT1}c'_{1,k}(y^{G*}_{1,k})-P_{1,i}&\geq 0\\
\label{GENPKKT2}y^{G*}_{1,k}(c'_{1,k}(y^{G*}_{1,k})-P_{1,i})&=0\\
\label{GENPKKT3}c'_{2,k}(y^{G*}_{2,k}(w)) - P_{2,i}(w)&\geq0, \quad\forall\,w\\
\label{GENPKKT4}y^{G*}_{2,k}(w)(c'_{2,k}(y^{G*}_{2,k}(w))-P_{2,i}(w))&=0, \quad\forall\,w.
\end{align}
In addition to the feasibility constraints (\ref{LSEP_const1})-(\ref{LSEP_const2}), the optimal solution for $(\text{LSE-P}_{i,k})$ satisfies 
\begin{align}
\label{LSEPKKT1}P_{1,i} - \sum_w\mu^{L*}_{2,i}(w)p_w&\geq 0\\
\label{LSEPKKT2}y^{L*}_{1,i}\left(P_{i,1} - \sum_w\mu^{L*}_{2,k}(w)p_w\right)&=0\\
\label{LSEPKKT3}P_{2,i}(w) - \mu^{L*}_{2,k}(w)&\geq 0, \quad\forall\,w\\
\label{LSEPKKT4}y^{L*}_{2,k}(w)\left(P_{2,i}(w) - \mu^{L*}_{2,k}(w)\right)&= 0, \quad\forall\,w\\
\label{LSEPKKT5}c'_{\text{dr},k}(x^{L*}_{2,k}(w)) - \mu^{L*}_{2,k}(w)&\geq 0,\quad \forall\,w\\
\label{LSEPKKT6}x^{L*}_{2,k}(w)\left(c'_{\text{dr},k}(x^{L*}_{2,k}(w)) - \mu^{L*}_{2,k}(w)\right)&=0, \quad\forall\,w\\
\label{LSEPKKT7}c'_{\text{bo},k}(z^{L*}_{2,k}(w)) - \mu^{L*}_{2,k}(w)&\geq 0,\quad \forall\,w\\
\label{LSEPKKT8}z^{L*}_{2,k}(w)\left(c'_{\text{bo},k}(z^{L*}_{2,k}(w)) - \mu^{L*}_{2,k}(w)\right)&=0, \quad\forall\,w\\
\label{LSEPKKT9}mu^{L*}_{2,k}(w)(D_k-w_k - y^{L*}_{1,k}- y^{L*}_{2,k}(w) - x^{L*}_{2,k}(w) - z^{L*}_{2,k}(w)))&=0, \,\,\forall\,w\\
\label{LSEPKKT10}\mu^{L*}_{2,k}(w)&\geq 0, \,\,\forall\,w,
\end{align}
where $\mu^{L*}_{2,k}(\cdot)$ is the optimal dual vector corresponding to constraint (\ref{LSEP_const1}) in $(\text{LSE-P}_{i,k})$. Now, we define candidate prices 
\begin{equation}
    \label{pricedef}P_{1,i} = \hat{\lambda}^*_{1,i}\quad \text{and}\quad P_{2,i}(w) = \hat{\lambda}^*_{2,i}(w), \quad\forall\,i,w,
\end{equation}
and claim the following. 
\begin{claim}
    $(\hat{y}^{G*}_1,\hat{y}^{L*}_1,\hat{y}^{G*}_2(\cdot),\hat{y}^{L*}_2(\cdot),\hat{x}^{L*}_2(\cdot),P_1,P_2(\cdot))$ is an SCEq, where $P_1$ and $P_2(\cdot)$ are defined in (\ref{pricedef}), and 
$(\hat{y}^{G*}_1,\hat{y}^{L*}_1,\hat{y}^{G*}_2(\cdot),\hat{y}^{L*}_2(\cdot),\hat{x}^{L*}_2(\cdot))$ is the unique solution to (SPP-P). 
\end{claim}
\begin{proof}
Starting with $(\text{GEN-P}_{i,k})$, substituting for $P_{1,i}$ and $P_{2,i}(w)$, and selecting $y^G_{1,k} = \hat{y}^{G*}_{1,k}$ and $y^G_{2,k}(w) = \hat{y}^{G*}_{2,k}(w)$ for all $w$ in (\ref{GENPKKT1})-(\ref{GENPKKT4}) yields expressions identical to (\ref{SPPKKT1})-(\ref{SPPKKT4}). This shows that, given $P_1$ and $P_2(\cdot)$ as defined in (\ref{pricedef}), $(\hat{y}^{G*}_{1,k},\hat{y}^{G*}_{2,k}(\cdot))$ is optimal for ($\text{GEN-P}_{i,k}$), and therefore $\hat{y}^{G*}_{1,k}$ is optimal for ($\text{GEN1}_{i,k}$) and $\hat{y}^{G*}_{2,k}(\cdot)$ for $(\text{GEN2}_{i,k})$. 

Continuing to the LSEs' problems, substituting for $P_1$ and $P_2(w)$ in (\ref{LSEPKKT1})-(\ref{LSEPKKT10}), and selecting $y^{L}_{1,k}=\hat{y}^{L*}_{1,k}$, $y^{L}_{2,k}(w)=\hat{y}^{L*}_{2,k}(w)$, $x^L_{2,k}(w)=\hat{x}^{L*}_{2,k}(w)$, $z^L_{2,k}(w) = \hat{z}^{L*}_{2,k}(w)$ for all $w$ yields expressions which are the same as (\ref{SPPKKT5})-(\ref{SPPKKT13}) and (\ref{SPPKKT18}), except that $\mu^{L*}_{2,i}(w)$ appears instead of $\hat{\mu}^*_{2,k}(w)$. Therefore, setting $\mu^{L*}_{2,k}(w)=\hat{\mu}^*_{2,k}(w)$ for all $w$ makes (\ref{LSEPKKT1})-(\ref{LSEPKKT10}) identical to (\ref{SPPKKT5})-(\ref{SPPKKT13}) and (\ref{SPPKKT18}), showing that $(\hat{y}^{L*}_{1,k},\hat{y}^{L*}_{2,k}(\cdot),\hat{x}^{L*}_{2,k}(\cdot),\hat{z}^{L*}_{2,k}(\cdot))$ from the (SPP-P) optimal solution gives an optimal solution for $(\text{LSE-P}_{i,k})$. Thus, given prices $P_1$ and $P_2(\cdot)$ as defined in (\ref{pricedef}), $\hat{y}^{L*}_{1,k}$ is optimal for $(\text{LSE1}_{i,k})$, and $(\hat{y}^{L*}_{2,k}(w),\hat{x}^{L*}_{2,k}(w),\hat{z}^{L*}_{2,k}(w))$ is optimal for $(\text{LSE2}_{i,k})$, for all $w$, given $\hat{y}^{L*}_{1,k}$. \end{proof}

The market clearing condition is satisfied due to feasibility constraints (\ref{SPP_powerbal1}) and (\ref{SPP_powerbal2}). Therefore, the tuple in the claim is a sequential competitive equilibrium, and we have proven Theorem 1. 
\end{proof}

\subsection{Social Welfare Theorems}\label{sec:welfare}

There exists an important connection between competitive equilibria and efficient allocations, described by the two fundamental theorems of welfare economics. Here we give statements of first and second theorems of welfare economics for our two-stage market setting. If an allocation is included in an SCEq, we say that the equilibrium \emph{supports} the allocation. 
\begin{theorem}
\emph{(i)} Every sequential competitive equilibrium supports an efficient sequential allocation. \emph{(ii)} Conversely, an efficient sequential allocation can be supported by a sequential competitive equilibrium. 
\end{theorem}
\begin{proof}
    To prove statement \emph{(i)}, per Definition 1, under a sequential competitive equilibrium, the market clears both in the DA and RT stage, in all scenarios. This condition is equivalent to posing the following ISO problem \cite{wang2012dynamic}:
\begin{align}
(\text{ISO})\max_{\substack{\tilde{y}^G_1,\tilde{y}^G_2,\tilde{y}^L_1,\tilde{y}^L_2\\\tilde{\theta}_1,\tilde{\theta}_2}}&\,\,\sum_iP_{1,i}\left(\sum_{k\in\mathcal{L}_i}\tilde{y}^L_{1,k}-\sum_{k\in\mathcal{G}_i}\tilde{y}^G_{1,k}\right)+\sum_w\left(\sum_iP_{2,i}(w)\left(\sum_{k\in\mathcal{L}_i}\tilde{y}^L_{2,k}(w)-\tilde{y}^G_{2,k}(w)\right)\right)p_w\\
\label{ISO_powerbal1}\text{s.t.}&\quad \sum_{k\in\mathcal{G}_i}\tilde{y}^G_{1,k} - \sum_{k\in\mathcal{L}_i}\tilde{y}^L_{1,k} = \sum_jB_{ij}(\tilde{\theta}_{1,i}-\tilde{\theta}_{1,j}), \,\,\forall\,i\\
\label{ISO_powerbal2}&\quad \sum_{k\in\mathcal{G}_i}\tilde{y}^G_{2,k}(w) - \sum_{k\in\mathcal{L}_i}\tilde{y}^L_{2,k}(w) =\sum_jB_{ij}(\tilde{\theta}_{2,i}(w) - \tilde{\theta}_{2,j}(w) - \tilde{\theta}_{1,i} + \tilde{\theta}_{1,j}), \,\,\forall\,i,w\\
\label{ISO_powercap1}&\quad B_{ij}(\tilde{\theta}_{1,i}-\tilde{\theta}_{1,j})\leq f^{\max}_{ij}, \quad\forall\,i,j\\
\label{ISO_powercap2}&\quad B_{ij}(\tilde{\theta}_{2,i}(w)-\tilde{\theta}_{2,j}(w))\leq f^{\max}_{ij}, \quad\forall\,i,j,w\\
&\quad \tilde{y}^G_{1,k}\geq0,\,\tilde{y}^G_{2,k}(w)\geq 0, \quad\forall\,i,k\in\mathcal{G}_i,w\\
&\quad \tilde{y}^L_{1,k}\geq0,\,\tilde{y}^L_{2,k}(w)\geq 0, \quad\forall\,i,k\in\mathcal{L}_i,w,
\end{align}
and then requiring that $(\overline{y}^{G*}_1,\overline{y}^{L*}_1,\overline{y}^{G*}_2(\cdot),\overline{y}^{L*}_2(\cdot))$ as given in the SCEq definition also solves (ISO). 

Summing the objectives of all agents, i.e.,  $(\text{GEN-P}_{i,k})$ and $(\text{LSE-P}_{i,k})$ over $(i,k)$, along with the objective of (ISO) recovers the objective of (SPP-P). Similarly, collecting the constraints from individual $(\text{GEN-P}_{i,k})$ and $(\text{LSE-P}_{i,k})$ problems, along with those from (ISO) recovers the full set of constraints found in (SPP-P). Together, therefore, (ISO), along with all $(\text{GEN-P}_{i,k})$ and $(\text{LSE-P}_{i,k})$ represent a decomposition of (SPP-P). 

Denote the Lagrange multipliers corresponding to constraints (\ref{ISO_powerbal1}) and (\ref{ISO_powerbal2}) as $\tilde{\lambda}_1$ and $\tilde{\lambda}_2$. Note that the KKT conditions corresponding to $\tilde{y}^G_{1,k}$ and $\tilde{y}^L_{1,k}$ are 
\begin{align}
    P_{1,i} - \tilde{\lambda}^*_{1,i}&\geq 0, \quad\forall\,i,k\in\mathcal{G}_i\\
    \tilde{y}^{G*}_{1,k}(P_{1,i} - \tilde{\lambda}^*_{1,i})&= 0, \quad\forall\,i,\,k\in\mathcal{G}_i\\
    \tilde{\lambda}^*_{1,i} - P_{1,i}&\geq 0, \quad\forall\,i,k\in\mathcal{G}_i\\
    \tilde{y}^{L*}_{1,k}( \tilde{\lambda}^*_{1,i} - P_{1,i})&= 0, \quad\forall\,i,\,k\in\mathcal{G}_i
\end{align}
This implies that we can take $\tilde{\lambda}^*_{1,i}=P_{1,i}$ for all $i$, and it can similarly be shown that $\tilde{\lambda}^*_{2,i}(w)=P_{2,i}(w)$ for all $i$ and $w$. The remaining KKT conditions for (ISO) are identical in form to (SPP-P) KKT conditions (\ref{SPPKKT14})-(\ref{SPPKKT17}). Therefore, associating Lagrange multipliers $\tilde{\gamma}^*_1$ and $\tilde{\gamma}^*_2(\cdot)$ with (ISO) constraints (\ref{ISO_powercap1}) and (\ref{ISO_powercap2}), together with  $(P_1,P_2(\cdot),\tilde{\gamma}^*_1,\tilde{\gamma}^*_2(\cdot),\tilde{\theta}^*_1,\tilde{\theta}^*_2(\cdot))$ satisfy (\ref{SPPKKT14})-(\ref{SPPKKT17}), and therefore give a candidate $(\hat{\lambda}^*_1,\tilde{\lambda}^*_2(\cdot),\tilde{\gamma}^*_1,\tilde{\gamma}^*_2(\cdot),\hat{\theta}^*_1,\hat{\theta}^*_2(\cdot))$ in (SPP-P). Taking $\mu^{L*}_{2,k}(w)=\hat{\mu}^*_{2,k}(w)$ for all LSEs as in the proof of Theorem \ref{compeq_existence} and choosing the primal allocation quantities in (SPP-P) as the equilibrium quantities  (together with implied shortfall decisions $z^{L*}_{2,k}(w)$ for all LSEs in all scenarios), the (SPP-P) KKT conditions are satisfied. Therefore, the sequential competitive equilibrium supports an efficient allocation. 

The proof of the second statement follows directly from the constructive proof of Theorem \ref{compeq_existence}. 
\end{proof}

\section{Two-Stage Network Mechanism for Electricity Market with Renewable Generation}\label{sec:mechanism}

We showed in the proof of Theorem \ref{compeq_existence} that SCEq prices arise from the dual solution to (SPP-P). If we make the additional assumptions that all market participant cost functions can be finitely parametrized (for example taking quadratic form), and further that all market participants are non-strategic, the following market mechanism implements the SCEq, and clears the market following the RT stage:
\begin{enumerate}
\item Each generator and LSE $(i,k)$ submits parameters $(\xi_{1,k},\xi_{2,k})$ and $(\xi_{\text{dr},k},\xi_{\text{bo},k})$, respectively. 
\item The ISO solves (SPP-P), and announces DA prices $P^*_1=\hat{\lambda}^*_1$, along with RT price schedule $P^*_2(\cdot)=\hat{\lambda}^*_2(\cdot)$. 
\item Generator $(i,k)$ solves $(\text{GEN1}_{1,k})$ and LSE $(i,k)$ solves $(\text{LSE1}_{1,k})$. LSE $(i,k)$ pays $P^*_{1,i}\overline{y}^{L*}_{1,k}$ and generator $(i,k)$ receives $P^*_{1,i}\overline{y}^{G*}_{1,k}$. 
\item At the start of the RT stage, the renewable generation output $W=w$ is observed by both the generators and LSEs. Generator $(i,k)$ solves $(\text{GEN2}_{i,k})$, and LSE $(i,k)$ solves ($\text{LSE2}_{i,k}$). LSE $(i,k)$ pays $P^*_{2,i}(w)\overline{y}^{L*}_{2,k}(w)$, and generator $(i,k)$ receives $P^*_{2,i}(w)\overline{y}^{G*}_{2,k}(w)$. 
\item Generator $(i,k)$ produces $\overline{y}^{G*}_{1,k}+\overline{y}^{G*}_{2,k}(w)$, and LSE receives $\overline{y}^{L*}_{1,k} + \overline{y}^{L*}_{2,k}(w)$. 
\end{enumerate}

\section{Dynamic Economic Dispatch Game and Efficient Bids}\label{sec:game}

Previous sections assumed that the ISO has full knowledge of the cost functions associated with each generator and LSE. In this section, we relax that assumption, instead allowing market participants to report information related to their respective costs to the ISO, which it then uses to make dispatch decisions. Additionally, we allow that all entities may behave strategically, so that the submitted information may not reflect their true costs. In practice, the bid formats typically implemented in power markets are not expressive enough to capture the strictly convex cost functions we posed as in previous sections. For example, the California ISO uses 10-segment piecewise linear bids for supply-side bids \cite{caisobidformat}. Therefore, in this section we reformulate (SPP-P) as a dynamic economic dispatch (DED) game and study the outcomes of that game. 
\subsection{LSE Utility Functions and SPP-P Reformulation}

In the following development, we will assume that generators and LSEs submit linear bids for the cost of energy production and value of energy consumption, respectively, for both the first and second stages of the market. For the RT market, both types of agents are allowed to submit bids corresponding to each scenario $W=w$. 

Given the objective of (SPP-P), it would be natural to allow the LSEs to submit bids on demand response and blackout costs. However, in order to provide sufficient conditions for the existence of Nash equilibria for our market in later sections, it is necessary instead to work with equivalent LSE valuation functions, i.e.,  functions giving the benefit LSEs derive from consuming quantities of primary and ancillary energy. 

In essence, this is due to the fact that under our market formulation, the ISO can only allocate demand response or planned blackouts to a given LSE $(i,k)$'s consumer population via the offerings of LSE $(i,k)$ itself. This is true even if multiple LSEs exist at a given node, or if multiple LSEs exist in a network with no congestion. In contrast, to the extent allowed by the network transmission constraints, the planner can route the least expensive electricity to serve loads. As we will show in a later example, if LSEs directly bid on their own service costs, it creates opportunities to increase their payoffs (decrease overall operating costs). On the other hand, when LSEs provide bids on their valuation for electricity in both stages, they compete directly with one another for the same service, and the planner is able to allocate flows to the highest bidders.

To start, consider the following problem, given $y^L_{1,k}$ and $y^L_{2,k}(w)$ for all scenarios $W=w$
\begin{align}
\nonumber(\text{LSE}_{i,k}(y^L_{1,k},y^L_{2,k}))\,\, \min_{x_{2,k},z_{2,k}}&\quad \sum_w(c_{\text{dr},k}(x_{2,k}(w))+ c_{\text{bo},k}(z_{2,k}(w)))p_w\\
\nonumber\text{s.t.}&\quad y^L_{1,k} + y^L_{2,k}(w) + x_{2,k}(w) + z_{2,k}(w)\geq D_k-w_k, \,\,\forall\,w.
\end{align}
This problem can be decomposed into $|\mathcal{W}|$ convex problems $(\text{LSE}_{i,k}(y^L_{1,k},y^L_{2,k}(w),w)$, each corresponding to a single scenario $w$, where $|\mathcal{W}|$ gives the total number of possible second stage scenarios. Problem $(\text{LSE}_{i,k}(y^L_{1,k},y^L_{2,k}))$ has KKT conditions
\begin{align}
    \label{LSEutilKKT1}c'_{\text{dr},k}(x^*_{2,k}(w)) - \mu^*_{2,k}(w)&\geq 0, \quad\forall\,w\\
    \label{LSEutilKKT2}x^*_{2,k}(w)\left(c'_{\text{dr},k}(x^*_{2,k}(w)) - \mu^*_{2,k}(w)\right)&=0, \quad\forall\,w\\
    \label{LSEutilKKT3}c'_{\text{bo},k}(z^*_{2,k}(w)) - \mu^*_{2,k}(w)&\geq 0, \quad\forall\,w\\
    \label{LSEutilKKT4}z^*_{2,k}(w)\left(c'_{\text{bo},k}(z^*_{2,k}(w)) - \mu^*_{2,k}(w)\right)&=0, \quad\forall\,w\\
    \label{LSEutilKKT5}\mu^*_{2,k}(w)(D_k - w_k - y^L_{1,k} - y^L_{2,k}(w)-x^*_{2,k}(w) - z^*_{2,k}(w))&=0, \,\,\forall\,w\\
    \label{LSEutilKKT6}\mu^*_{2,k}(w)&\geq 0, \,\,\forall\,w. 
\end{align}
Let $u_{i,k}(y^L_{1,k},y^L_{2,k})$ denote the negation of the optimal value of $(\text{LSE}_{i,k}(y^L_{1,k},y^L_{2,k}))$, and $u_{i,k}(y_{1,k},y_{2,k}(w),w)$ denote the negation of the optimal value of $(\text{LSE}_{i,k}(y^L_{1,k},y^L_{2,k}),w)$ for each $w$ (without scaling by $p_w$). Thus these functions give the benefit of consuming electricity acquired in the first and second stages of the market in terms of the associated negative costs of demand response and scheduled blackouts. Note that the definitions of $u_{i,k}(y^L_{1,k},y^L_{2,k})$ and $u_{i,k}(y^L_{1,k},y^L_{2,k}(w),w)$ imply the following equality:
$$u_{i,k}(y^L_{1,k},y^L_{2,k}) = \sum_wu_{i,k}(y^L_{1,k},y^L_{2,k}(w),w)p_w.$$
As each $(\text{LSE}_{i,k}(y^L_{1,k},y^L_{2,k},w))$ is a convex problem, its optimal value function is convex in both $y^L_{1,k}$ and $y^L_{2,k}(w)$, implying that $u_{i,k}(\cdot,\cdot,w)$ is convex for all $w$ in both arguments \cite{boyd2004convex}. Since (\ref{LSEutilKKT1})-(\ref{LSEutilKKT6}) is a subset of the (SPP-P) KKT conditions, we can use the specified utility functions to rewrite (SPP-P) without explicit reference to $x^L_{2,k}(w)$ and $z^L_{2,k}(w)$:
\begin{align}
(\text{SPP-U})\max_{\substack{\hat{y}^G_1,\hat{y}^G_2,\hat{y}^L_1,\hat{y}^L_2\\\hat{\theta}_1,\hat{\theta}_2}}&\quad\sum_i\sum_{k\in\mathcal{L}_i}\sum_wu_{i,k}(y^L_{1,k},y^L_{2,k}(w))p_w-\sum_i\sum_{k\in\mathcal{G}_i}\left(c_{1,k}(\hat{y}^G_{1,k}) +\sum_{w}c_{2,k}(\hat{y}^G_{2,k}(w)p_w\right)\\
\label{SPPU_powerbal1}\text{s.t.}&\quad \sum_{k\in\mathcal{G}_i}\hat{y}^G_{1,k} - \sum_{k\in\mathcal{L}_i}\hat{y}^L_{1,k} = \sum_jB_{ij}(\hat{\theta}_{1,i}-\hat{\theta}_{1,j}), \,\,\forall\,i\\
\label{SPPU_powerbal2}&\quad \sum_{k\in\mathcal{G}_i}\hat{y}^G_{2,k}(w) - \sum_{k\in\mathcal{L}_i}\hat{y}^L_{2,k}(w)= \sum_jB_{ij}(\hat{\theta}_{2,i}(w)-\hat{\theta}_{2,j}(w)-\hat{\theta}_{1,i}+\hat{\theta}_{1,j}), \,\,\forall\,i,w\\
&\label{SPPU_branchflow1}\quad B_{ij}(\hat{\theta}_{1,i} - \hat{\theta}_{1,j})\leq f^{\max}_{ij}, \quad\forall\,i,j\\
&\label{SPPU_branchflow2}\quad B_{ij}(\hat{\theta}_{2,i}(w) - \hat{\theta}_{2,j}(w))\leq f^{\max}_{ij}, \quad\forall\,i,j\\
&\quad \hat{y}^L_{1,k}\geq 0,\hat{y}^L_{2,k}(w)\geq0, \quad\forall\,i,k\in\mathcal{L}_i,w\\
\label{SPPU_const_end}&\quad \hat{y}^G_{1,k}\geq0,\hat{y}^L_{2,k}(w)\geq 0, \quad\forall\,i,k\in\mathcal{G}_i,w
\end{align}
While $u_{i,k}(\cdot,\cdot,w)$ may not be differentiable under our original cost function assumptions in all cases, we can still make use of the subdifferential versions of the KKT conditions to describe optimal solutions to (SPP-U). 
\subsection{DED Game and Efficient Bids}
In the DED, each generator $(i,k)$ bids $b^G_{1,k}(\cdot)$ and $b^G_{2,k}(\cdot,w)$ for each scenario, representing its reported cost for producing electricity in both stages. Similarly, the LSEs bid $b^L_{1,k}(\cdot)$ and $b^L_{2,k}(\cdot,w)$ for each scenario, giving their utility for consuming primary and ancillary generation. Here it is assumed that all participants submit nonnegative scalar bids in each case, specifying cost and utility functions, e.g. $b^G_{1,k}(x) = b^G_{1,k}x$ for all $x\geq 0$. 

Given the generator and LSE bids, the ISO solves the following problem:
\begin{align}
\nonumber(\text{DED})\,\max_{\substack{\hat{y}^G_1,\hat{y}^G_2,\hat{y}^L_1,\hat{y}^L_2\\\hat{\theta}_1,\hat{\theta}_2}}&\,\sum_i\sum_{k\in\mathcal{L}_i}b^L_{1,k}\hat{y}^L_{1,k}+ \sum_i\sum_{k\in\mathcal{L}_i}\sum_wb^L_{2,k}(w)\hat{y}^L_{2,k}(w))p_w\\
&\hspace{0.3in}-\sum_i\sum_{k\in\mathcal{G}_i}\left(b^G_{1,k}\hat{y}^G_{1,k} +\sum_{w}b^G_{2,k}(w)\hat{y}^G_{2,k}(w)p_w\right)\\
\label{DED_powerbal1}\text{s.t.}&\quad \sum_{k\in\mathcal{G}_i}\hat{y}^G_{1,k} - \sum_{k\in\mathcal{L}_i}\hat{y}^L_{1,k} = \sum_jB_{ij}(\hat{\theta}_{1,i}-\hat{\theta}_{1,j}), \quad\forall\,i\\
\label{DED_powerbal2}&\quad \sum_{k\in\mathcal{G}_i}\hat{y}^G_{2,k}(w) - \sum_{k\in\mathcal{L}_i}\hat{y}^L_{2,k}(w)= \sum_jB_{ij}(\hat{\theta}_{2,i}(w)-\hat{\theta}_{2,j}(w)-\hat{\theta}_{1,i}+\hat{\theta}_{1,j}), \,\,\forall\,i,w\\
&\quad B_{ij}(\hat{\theta}_{1,i} - \hat{\theta}_{1,j})\leq f^{\max}_{ij}, \quad\forall\,i,j\\
&\quad B_{ij}(\hat{\theta}_{2,i}(w) - \hat{\theta}_{2,j}(w))\leq f^{\max}_{ij}, \quad\forall\,i,j\\
    &\quad \hat{y}^L_{1,k}\geq 0,\hat{y}^L_{2,k}(w)\geq0, \quad\forall\,i,k\in\mathcal{L}_i,w\\
    \label{DED_const_end}&\quad \hat{y}^G_{1,k}\geq0,\hat{y}^L_{2,k}(w)\geq 0, \quad\forall\,i,k\in\mathcal{G}_i,w
\end{align}
First, we show by construction that there exists a bid profile $b = (b^G_{1,k},b^G_{2,k},b^L_{1,k},b^L_{2,k})$ for (DED) that induces an allocation which is efficient, i.e.,  optimal for (SPP-P). 
\begin{proposition}\label{efficient_bid_prop}
For each generator $(i,k)$, let 
\begin{equation}
\label{efficient_bid}
b^G_{1,k} = b^L_{1,k} = \lambda^*_{1,i},\quad b^G_{2,k}(w) = b^L_{2,k}(w) = \lambda^*_{2,i}(w), \quad\forall\,i,w,\end{equation}
where $\lambda^*_{1,i}$ and $\lambda^*_{2,i}(w)$ give optimal Lagrange multipliers associated with (SPP-P) constraints (\ref{SPP_powerbal1}) and (\ref{SPP_powerbal2}) for all $i$ and $w$. Then, there exists a solution to (DED) that is also a solution to (SPP-P). 
\end{proposition}
\begin{proof}
In addition to feasibility, the KKT conditions for (DED) are:
\begin{align}
    \label{DEDKKT1}b^G_{1,k} - \hat{\lambda}^*_{1,i}&\geq 0, \quad\forall\,i,k\in\mathcal{G}_i\\
    \label{DEDKKT2}\hat{y}^{G*}_{1,k}\left(b^G_{1,k} - \hat{\lambda}^*_{1,i}\right)&=0, \quad\forall\,i,k\in\mathcal{G}_i\\
    \label{DEDKKT3}b^G_{2,k}(w) - \hat{\lambda}^*_{2,i}(w)&\geq 0, \quad\forall\,i,k\in\mathcal{G}_i\\
    \label{DEDKKT4}\hat{y}^{G*}_{2,k}(w)\left(b^G_{2,k}(w) - \hat{\lambda}^*_{2,i}(w)\right)&=0, \quad\forall\,i,k\in\mathcal{G}_i\\
    \label{DEDKKT5}\hat{\lambda}^*_{1,i} - b^L_{1,k}&\geq 0, \quad\forall\,i,k\in\mathcal{G}_i,w\\
    \label{DEDKKT6}\hat{y}^{L*}_{1,k}\left(\hat{\lambda}^*_{1,i} - b^L_{1,k}\right)&=0, \quad\forall\,i,k\in\mathcal{G}_i,w\\
    \label{DEDKKT7}\hat{\lambda}^*_{2,i}(w)-b^L_{2,k}(w)&\geq 0, \quad\forall\,i,k\in\mathcal{L}_i,w\\
    \label{DEDKKT8}\hat{y}^{L*}_{2,k}(w)\left(\hat{\lambda}^*_{2,i}(w)-b^L_{2,k}(w)\right)&=0, \quad\forall\,i,k\in\mathcal{L}_i,w\\
   \label{DEDKKT9}\sum_jB_{ij}\big(\hat{\lambda}^*_{1,i} - \hat{\lambda}^*_{1,j}-\sum_w\big(\hat{\lambda}^*_{2,i}(w)-\hat{\lambda}^*_{2,j}(w)\big)p_w+ \hat{\gamma}^*_{1,ij}-\hat{\gamma}^*_{1,ji}\big)&=0, \,\,\forall\,i\\
     \label{DEDKKT10} \sum_jB_{ij}\big(\hat{\lambda}^*_{2,i}(w) - \hat{\lambda}^*_{2,j}(w)
   + \hat{\gamma}^*_{2,ij}(w) - \hat{\gamma}^*_{2,ij}(w)\big)&=0, \,\,\forall\,i,w\\
    \label{DEDKKT11}\hat{\gamma}^*_{1,ij}\left(B_{ij}(\hat{\theta}^*_{1,i}-\hat{\theta}^*_{1,j}-f^{\max}_{ij}\right)&=0, \,\,\forall\,i,j\\
    \label{DEDKKT12}\hat{\gamma}^*_{2,ij}\left(B_{ij}(\hat{\theta}^*_{2,i}(w)-\hat{\theta}^*_{2,j}(w)-f^{\max}_{ij}\right)&=0, \,\,\forall\,i,j\\
    \label{DEDKKT13}\hat{\gamma}^*_{1,ij}\geq 0,\hat{\gamma}^*_{2,ij}(w)&\geq0, \,\,\forall\,i,j,w
\end{align}
Under the bid profile given in (\ref{efficient_bid}), the planner can select $\hat{\lambda}^*=\lambda^*$ to satisfy conditions (\ref{DEDKKT1})-(\ref{DEDKKT8}). Given this selection, the planner can select $(\hat{\theta}^*,\hat{\gamma}^*)=(\theta^*,\gamma^*)$ to satisfy the remaining conditions, and then select $\hat{y}^*=y^*$, where $(y^*,\theta^*,\gamma^*)$ are part of an optimal solution to (SPP-U). 
\end{proof}

\section{Sequential Nash Equilibria}\label{sec:sceq}

While the previous section demonstrated the existence of an efficient bid profile for the dynamic economic dispatch game, given that the market participants are free to bid any nonnegative scalar values, it remains to show whether generators or LSEs might find it in their own interest to deviate from such a profile. 

To begin, we define individual outcomes, i.e.,  the payoffs for the generators and LSEs under the DED with linear bids, assuming the locational marginal pricing (LMP) scheme \cite{morales2013integrating} is enforced by the ISO. For a given bid profile $b$, and the resulting DED solution allocation $(\hat{y}^{G*}_1(b),\hat{y}^{G*}_2(b),\hat{y}^{L*}_1(b),\hat{y}^{L*}_2(b))$, the expected payoff for generator $(i,k)$ is 
\begin{equation}\nonumber
\begin{split}
\mathbb{E}[\pi^G_{i,k}(b)] = \hat{\lambda}^*_{1,i}(b)\hat{y}^{G*}_{1,i}(b) - c_{1,k}(\hat{y}^{G*}_{1,i}(b))+ \sum_w\left(\hat{\lambda}^*_{2,i}(b,w)\hat{y}^{G*}_{2,i}(b,w) - c_{2,k}(\hat{y}^{G*}_{2,k}(b,w))\right)p_w. 
\end{split}
\end{equation}
The expected payoff for LSE $(i,k)$ is 
\begin{equation}\nonumber
\begin{split}
&\mathbb{E}[\pi^L_{i,k}(b)] = -\hat{\lambda}^*_{1,i}(b)\hat{y}^{L*}_{1,k}(b)+ \sum_wu_{i,k}(\hat{y}^{L*}_{1,k}(b),\hat{y}^{L*}_{2,k}(b,w),w)p_w-\sum_w\hat{\lambda}^*_{2,k}(b,w)\hat{y}^{L*}_{2,k}(b,w))p_w. 
\end{split}
\end{equation}
For a bid profile $b$, denote the collection of bids aside from $b_{i,k}$ as $b_{-(i,k)}$. Through the remainder of this section, we will omit the expectation notation $\mathbb{E}[\cdot]$ and denote the expected payoff of individual generators and LSEs, given bid profile $b$ as $\pi^G_{i,k}(b)$ and $\pi^L_{i,k}(b)$, respectively. 
\begin{definition}
A \emph{sequential Nash equilibrium} is a bid profile $b^*$ such that for all generators, it holds that 
\begin{equation}
\label{nash_eq_gen}
\pi^G_{i,k}(b^*)\geq \pi^G_{i,k}(b_{i,k},b_{-(i,k)})\quad \forall\,b_{i,k}\in\mathbb{R}_+\times\mathbb{R}^{|\mathcal{W}|}_+
\end{equation}
and for all LSEs it holds that 
\begin{equation}
\label{nash_eq_lse}
\pi^L_{i,k}(b^*)\geq \pi^L_{i,k}(b_{i,k},b_{-(i,k)})\quad \forall\,b_{i,k}\in\mathbb{R}_+\times\mathbb{R}^{|\mathcal{W}|}_+.
\end{equation}
\end{definition}
\subsection{Existence of Efficient Sequential Nash Equilibria}
This section explores two conditions under either of which an efficient Nash equilibrium exists for (DED). For both conditions, the efficient bid profile specified in the previous section coincides with such a Nash equilibrium. First, the following assumption is necessary to preclude the possibility that any single generator has the market power to ask for arbitrarily high prices in its bid.
\begin{assumption}
\label{exclude_gen}
The system problem (SPP-U) is feasible in the absence of any one of the generators in either market stage. 
\end{assumption}
The first sufficient condition is as follows. 
\begin{definition}
\emph{(Congestion Free Condition)} No branch power flow constraint (\ref{SPPU_branchflow1}) or (\ref{SPPU_branchflow2}) is binding in the optimal dispatch for (SPP-U). 
\end{definition}

With the previous section's existence result in hand, it is possible to show that the congestion-free condition guarantees the existence of an efficient Nash equilibrium. 
\begin{theorem}
Under Assumption \ref{exclude_gen} and the congestion-free condition, there exists an efficient Nash equilibrium for (DED). 
\end{theorem}
\begin{proof}
Considering bid profile $b^*$ as in (\ref{efficient_bid}). Note that under the congestion-free condition, the bid profile of all generators and LSEs are identical, as the LMPs become uniform across all nodes in the network \cite{morales2013integrating}. Denote these uniform prices as $\lambda^*_1(b^*)$ and $\lambda^*_2(b^*,w)$ for all $w$.  By Proposition \ref{efficient_bid_prop}, the bid profile is efficient, inducing optimal dispatch $\{y^{G*}_1,y^{L*}_1,y^{G*}_2(\cdot),y^{L*}_2(\cdot),\theta^*_1,\theta^*_2(\cdot),\gamma^*_1,\gamma^*_2(\cdot)\}$. 

To show that this bid profile is an efficient Nash equilibrium, we start with the generators. Collecting constraints (\ref{DEDKKT1}) and (\ref{DEDKKT3}) across all nodes and generators gives that 
\begin{align}
\hat{\lambda}^*_{1,i}(b^*)\leq \min_{j,k\in\mathcal{G}_j}\,b^{G*}_{1,k} &= \lambda^*_1(b^*)\\
\hat{\lambda}^*_{2,i}(b^*,w)\leq \min_{j,k\in\mathcal{G}_j}\,b^{G*}_{2,k}(w) &= \lambda^*_2(b^*,w).
\end{align}
Note that when \emph{any} electricity is generated anywhere in the network for a given stage, the inequalities above are tight. If a generator $(i,k)$ deviates in its bid from $b^{G*}_{i,k}$ to $b^G_{i,k}$, then these inequalities become 
\begin{align}
\label{gen_deviation_ineq1}\hat{\lambda}^*_{1,i}\left(b^G_{i,k},b^*_{-(i,k)}\right)\leq &\min\left\{b_{1,k},\lambda^*_1(b^*)\right\}\\
\label{gen_deviation_ineq2}\hat{\lambda}^*_{2,i}\left(b^G_{2,k},b^*_{-(i,k)}\right)\leq &\min\left\{b_{2,k}(w),\lambda^*_2(b^*,w)\right\}
\end{align}
Inequalities (\ref{gen_deviation_ineq1}) and (\ref{gen_deviation_ineq2}) show that the generator can only bring its nodal price down by deviating in its bid. Setting $b = \left(b^G_{i,k},b^*_{-(i,k)}\right)$, its expected profit when deviating unilaterally becomes 
\begin{align}
\nonumber\pi^G_{i,k}(b) &= \hat{\lambda}^*_{1,i}(b)\hat{y}^{G*}_{1,k}(b) - c_{1,k}(\hat{y}^{G*}_{1,k}(b))+ \sum_w\left(\hat{\lambda}^*_{2,i}(b,w)\hat{y}^{G*}_{2,k}(b,w) - c_{2,k}(\hat{y}^{G*}_{2,k}(b,w))\right)p_w\\
\label{gen_max_target}&\leq \hat{\lambda}^*_{1,i}(b^*)\hat{y}^{G*}_{1,k}(b) - c_{1,k}(\hat{y}^{G*}_{1,k}(b)) + \sum_w\left(\hat{\lambda}^*_{2,k}(b^*,w)\hat{y}^{G*}_{2,k}(b,w) - c_{2,k}(\hat{y}^{G*}_{2,k}(b,w))\right)p_w
\end{align}
Given $\hat{\lambda}^*_{1,i}(b^*)$ and $\hat{\lambda}^*_{2,i}(b^*,w)$ for all $w$, (SPP-P) KKT conditions (\ref{SPPKKT1})-(\ref{SPPKKT4}) constitute first order conditions for maximization of the right hand side of (\ref{gen_max_target}). This implies that out of all possible allocations, $\{\hat{y}^{G*}_{1,k}(b^*),\hat{y}^{G*}_{2,k}(b^*,w)\}$ maximizes the right hand side of (\ref{gen_max_target}), and 
\begin{align*}
\nonumber\pi^G_{i,k}(b) &\leq \hat{\lambda}^*_{1,i}(b^*)\hat{y}^{G*}_{1,k}(b^*) - c_{1,k}(\hat{y}^{G*}_{1,k}(b^*))+ \sum_w\left(\hat{\lambda}^*_{2,i}(b^*,w)\hat{y}^{G*}_{2,k}(b^*,w) - c_{2,k}(\hat{y}^{G*}_{2,k}(b^*,w))\right)p_w\\
&= \pi^G_{i,k}(b^*).
\end{align*}
Therefore, no generator can increase its payoff by unilaterally deviating from $b^{G*}_{i,k}$ in its bid.

Turning to the LSEs, collecting (DED) KKT conditions (\ref{DEDKKT5}) and (\ref{DEDKKT7}) across all nodes and LSEs gives that
\begin{align}
\lambda^*_1(b^*) = \max_{j,k\in\mathcal{L}_j}\,b^{L*}_{1,k}&\leq \hat{\lambda}^*_{1,i}\\
\lambda^*_2(b^*,w) = \max_{j,k\in\mathcal{L}_j}\,b^{L*}_{2,k}(w)&\leq \hat{\lambda}^*_{2,i}(w).
\end{align}
If an LSE $(i,k)$ deviates in its bid from $b^{L*}_{i,k}$ to $b^L_{i,k}$, then these inequalities become 
\begin{align}
\max\left\{b^L_{1,k},\lambda^*_1\right\}&\leq \hat{\lambda}^*_{1,i}\\
\max\left\{b^L_{2,k}(w),\lambda^*_2(b^*,w)\right\}&\leq \hat{\lambda}^*_{2,i}(w).
\end{align}
Thus, any LSE $(i,k)$ can only drive their nodal electricity prices up when deviating from $b^{L*}_{i,k}$. Setting $b=(b^L_{i,k},b^*_{-(i,k)})$, its expected payoff when deviating unilaterally becomes 
\begin{align}
\nonumber\pi^L_{i,k}(b) &= -\hat{\lambda}^*_{1,i}(b)\hat{y}^{L*}_{1,k}(b)+ \sum_w(u_{i,k}(\hat{y}^{L*}_{1,k}(b),\hat{y}^{L*}_{2,k}(b,w),w)p_w- \sum_w\hat{\lambda}^*_{2,k}(b,w)\hat{y}^{L*}_{2,k}(b,w))p_w\\
\label{lse_max_target} &\leq -\hat{\lambda}^*_{1,i}(b^*)\hat{y}^{L*}_{1,k}(b)+\sum_wu_{i,k}(\hat{y}^{L*}_{1,k}(b),\hat{y}^{L*}_{2,k}(b,w),w)p_w -\sum_w\hat{\lambda}^*_{2,k}(b^*,w)\hat{y}^{L*}_{2,k}(b,w))p_w.
\end{align}
The KKT conditions for (SPP-U) concerning the LSE utility function terms are as follows:
\begin{align*}
    \hat{\lambda}^*_{1,i} - \sum_wu'_{i,k}(\hat{y}^{L*}_{1,k},\hat{y}^{L*}_{2,k}(w),w)p_w&\geq 0, \,\,\forall\,i,k\in\mathcal{L}_i\\
    \hat{y}^{L*}_{1,k}\left(\hat{\lambda}^*_{1,i} - \sum_wu'_{i,k}(\hat{y}^{L*}_{1,k},\hat{y}^{L*}_{2,k}(w),w)p_w\right)&\geq 0, \,\,\forall\,i,k\in\mathcal{L}_i,
    \end{align*}
    and for all $w$
    \begin{align*}
    \hat{\lambda}^*_{2,i}(w) - u'_{i,k}(\hat{y}^{L*}_{1,k},\hat{y}^{L*}_{2,k}(w),w)&\geq 0, \,\,\forall\,i,k\in\mathcal{L}_i\\
    \hat{y}^{L*}_{2,k}(w)\left(\hat{\lambda}^*_{2,i}(w) - u'_{i,k}(\hat{y}^{L*}_{1,k},\hat{y}^{L*}_{2,k}(w),w)\right)&\geq 0, \,\,\forall\,i,k\in\mathcal{L}_i.
\end{align*}
Note that since it is the sum of $\hat{y}^L_{1,k}$ and $\hat{y}^L_{2,k}(w)$ that appears in $\text{LSE}_{i,k}(\hat{y}^L_{1,k},\hat{y}^L_{2,k}(w),w)$, the set of subgradients of $u_{i,k}(\cdot,\cdot,w)$ with respect to $\hat{y}^L_{1,k}$ or $\hat{y}^{L}_{2,k}(w)$ is the same, and wherever $u_{i,k}(\cdot,\cdot,w)$ is differentiable, the derivative with respect to either argument is the same.

As in the proof for the generators, these conditions show that given $\hat{\lambda}^*_{1,i}$ and $\hat{\lambda}_{2,i}(w)$ for all $w$, out of all possible allocations, $\{\hat{y}^{L*}_{1,k}(b^*),\hat{y}^{L*}_{2,k}(b^*,w)\}$ maximizes the right hand side of (\ref{lse_max_target}), and 
\begin{align*}
\pi^L_{i,k}(b)&\leq -\hat{\lambda}^*_{1,i}(b^*)\hat{y}^{L*}_{1,k}(b^*)+ \sum_wu_{i,k}(\hat{y}^{L*}_{1,k}(b^*),\hat{y}^{L*}_{2,k}(b^*,w),w)p_w- \sum_w\hat{\lambda}^*_{2,k}(b^*,w)\hat{y}^{L*}_{2,k}(b^*,w))p_w= \pi^L_{i,k}(b^*). 
\end{align*}
Therefore, no LSE can increase its payoff by unilaterally deviating from $b^{L*}_{i,k}$ in its bid. 
\end{proof}
The second sufficient condition for existence of an efficient Nash equilibrium for (DED) is the following. 
\begin{definition}
    \emph{(Monopoly-Free Condition)} At each bus, there are either at least two generators, or no generators at all. Also, at each bus,  there are either at least two LSEs, or no LSEs at all. 
\end{definition}

Before proving the sufficiency of the monopoly-free condition for ensuring the existence of an efficient Nash equilibrium, we provide a counter example to demonstrate the necessity of our reformulation of the (SPP-P) problem in terms of LSE utility functions. The following example demonstrates that if LSEs bid marginal costs for demand response, then the monopoly-free condition fails to ensure the existence of an efficient Nash equilibrium. 
\begin{example} For this counter example, we set aside planned blackouts, and uncertainty in generation, so that (SPP-P) reduces to a single stage problem, and consider the following network, illustrated in Figure \ref{networkfig}.
\begin{figure}[h]
\begin{center}
\includegraphics[width=3.45in]{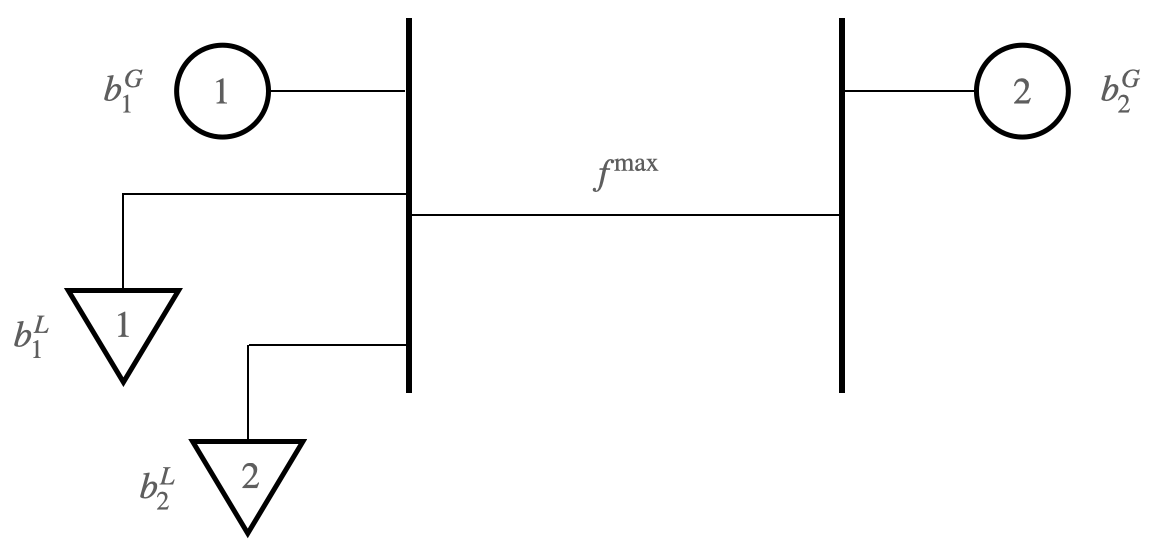}
\caption{Network diagram for Example 1. Generators are represented by circles, LSEs are represented by triangles.}
\label{networkfig}
\end{center}
\end{figure}
\end{example}
Suppose that both the generator cost functions and LSE demand response cost functions are quadratic, with no constant terms. There are two buses, the first with a single generator and two LSEs, the second with a single generator. For this example the generators are assumed to be nonstrategic, as we focus on the strategic behavior of the LSEs. We simplify the subscripts here to refer only to specific generators and LSEs: 
\begin{align*}
c_1(y) &= 80(y)^2+40y\\
c_2(y) &= 40(y)^2+20y\\
c_{\text{dr},1}(x) &= 10(x)^2+20x,\quad D_1 = 30\\
c_{\text{dr},2}x) &= 10(x)^2 + 30x,\quad D_2=20
\end{align*}
For the network characteristics, we take $f^{\max}_{12}=f^{\max}_{21} = 2$, and $B_{12}=B_{21}=1$. The optimal primal solution for (SPP-U) with these parameters is as follows:
\begin{gather}
\label{example_primal_sol1}\hat{y}^{G*}_1 = 3,\quad \hat{y}^{G*}_2 = 2,\quad \hat{\theta}^* = 2\\
\label{example_primal_sol2} (\hat{x}^{L*}_1,\hat{y}^{L*}_1) = (25,5),\quad (\hat{x}^{L*}_2,\hat{y}^{L*}_2) = (30,0) 
\end{gather}
with optimal dual solution 
\begin{gather}
\label{example_dual_sol1}\hat{\lambda}^*_1 = 520,\quad \hat{\lambda}^*_2 = 180,\quad \hat{\mu}^*_1 = 520,\quad \hat{\mu}^*_2 = 430,\\
\label{example_dual_sol2}\hat{\gamma}^*_{12} = 0,\quad \hat{\gamma}^*_{21} = 340. 
\end{gather}
If the generators each bid the LMP for their respective node, and the LSEs bid $\hat{\mu}^*$, corresponding to their respective marginal costs for demand response at the optimal allocation (\ref{example_primal_sol1})-(\ref{example_primal_sol2}), i.e.,  if each entity bids their respective term from (\ref{example_dual_sol1})-(\ref{example_dual_sol2}), then the (SPP-P) primal and dual solutions are optimal for the corresponding (DED) game.  

It can be shown that LSE 1 will not benefit from increasing its bid, which could only cause an increase in $\hat{\lambda}^*_1$, the price it pays for electricity. However, due to the KKT conditions of the (DED) game, it can unilaterally deviate and bid in the range 
\begin{equation}
b^{L*}_2 = \hat{\mu}^*_2(b^*) = 430 \leq b^L_1 \leq 520 = \hat{\lambda}^*_1(b^*) = b^{G*}_1,
\end{equation}
where the dependence of the (DED) dual variables on the bid profile is made explicit, and $b^*$ is the efficient bid profile given in (\ref{efficient_bid}). Under bid profile $b^*$, the payoff for LSE 1 is 
\begin{align}
\pi^L_1(b^*) &= -\hat{\lambda}^*_1(b^*)\hat{y}^{L*}_1(b^*) - c_{\text{dr},1}(\hat{x}^{L*}_1(b^*))\\
&= -520\cdot 5 - 10\cdot 25^2 - 20\cdot25\\
&= -9350.
\end{align}
Suppose that LSE 1 deviates and bids $b^L_1 = 440$. Then, denoting the bid including LSE 1's deviation as $b$, the optimal (DED) primal solution becomes 
\begin{gather*}
\hat{y}^{G*}_1(b) = 0,\quad \hat{y}^{G*}_2(b) = 2,\quad \hat{\theta}^*(b) = 2\\
 (\hat{x}^{L*}_1(b),\hat{y}^{L*}_1(b)) = (28,2),\quad (\hat{x}^{L*}_2(b),\hat{y}^{L*}_2(b)) = (30,0). 
\end{gather*}
Note that since the reported cost of generator 2 remains lower than the cost of demand response for either LSE, the planner still finds it optimal to dispatch generator 2 to the extent that the transmission line constraint $f^{\max}_{12} = 2$ allows. Since both LSEs now underbid generator 1, generator 1 is not dispatched. The two units provided by generator 2 go to LSE 1, since it still reports a higher demand response cost than LSE 2, and LSE 1's reported demand response cost now sets the price of energy at node 1. The corresponding dual solution is 
\begin{gather*}
\hat{\lambda}^*_1(b) = 440,\quad \hat{\lambda}^*_2(b) = 180,\quad \hat{\mu}^*_1(b) = 440,\quad \hat{\mu}^*_2(b) = 430\\
\hat{\gamma}^*_{12}(b) = 0,\quad \hat{\gamma}^*_{21}(b) = 260, 
\end{gather*}
and LSE 1's payoff under bid profile $b'$ is 
\begin{align}
\pi^L_1(b) &= -\hat{\lambda}^*_1(b)\hat{y}^{L*}_1(b) - c_{\text{dr},1}(\hat{x}^{L*}_1(b))\\
&= -440\cdot 2 - 10\cdot 28^2 - 20\cdot28\\
&= -9280> -9350 =  \pi^L_1(b^*).
\end{align}
Thus, while LSE 1 consumes less energy at the reduced price, this reduction in cost of energy consumption outweighs the increased cost of demand response it must take on in the outcome under demand profile $b$. \\
Note that if LSE 2 were to also consume energy, then its bid $b^{L*}_2$ would be equal to 520 as well, and LSE 1 would not be able to bring the cost of electricity at node 1 down with its demand response bid. In essence, our reformulation of (SPP-P) in terms of LSE utility functions for electricity consumption, together with the prescribed bid format and efficient bid (\ref{efficient_bid}) makes the monopoly-free a sufficient condition because it ensures that for each LSE in the network at any given node, there will always exist another LSE bidding the equilibrium price at the same node. Therefore, LSEs which underbid will be forced to fulfill their demand entirely through demand response or planned blackouts, which can be shown to yield a payoff no better than their equilibrium payoff. The details can be found in the following proof. 

\begin{theorem}
	Under Assumption \ref{exclude_gen} and the monopoly-free condition, there exists an efficient Nash equilibrium for (DED). 
\end{theorem}
\begin{proof}
The proof for the sufficiency of the monopoly-free condition is similar to that of the congestion-free condition. Since there is no guarantee on congestion conditions, the optimal dual variables corresponding to the (SPP-U) power balance constraints (\ref{SPPU_powerbal1}) and (\ref{SPPU_powerbal2}) are $\lambda^*_1 = (\lambda^*_{1,1},\dots,\lambda^*_{1,N})^{\top}$ and $\lambda^*_{2,i}(w) = (\lambda^*_{2,1}(w),\dots,\lambda^*_{2,N}(w))^{\top}$ for all $i$. In general the individual entries of these vectors may take on different values.

Under the monopoly-free condition, the bounds on the nodal prices, i.e.,  the Lagrange multipliers corresponding to (DED) power balance constraints (\ref{DED_powerbal1})-(\ref{DED_powerbal2}) can only be aggregated per node. In this case of the generators, for an individual node $i$ this gives
\begin{align}
\hat{\lambda}^*_{1,i}(b^*)\leq \min_{k\in\mathcal{G}_i}\,b^{G*}_{1,k} &= \lambda^*_{1,i}(b^*)\\
\hat{\lambda}^*_{2,i}(b^*,w)\leq \min_{k\in\mathcal{G}_i}\,b^{G*}_{2,k}(w) &= \lambda^*_{2,i}(b^*,w).
\end{align}
Under unilateral deviation by generator $(i,k)$ to bid $b^G_{i,k}$, these inequalities become 
\begin{align}
\hat{\lambda}^*_{1,i}\left(b^G_{i,k},b^*_{-(i,k)}\right)&\leq \min\left\{ b^G_{1,k},\lambda^*_{1,i}(b^*)\right\}\\
\hat{\lambda}^*_{2,i}(b^G_{i,k},b^*_{-(i,k)},w)&\leq  \min\left\{ b^G_{2,k}(w),\lambda^*_{2,i}(b^*,w)\right\}.
\end{align}
Again, the nodal prices can only decrease as a result of the generator's unilateral deviation, so that the remainder of the previous proof can be used in this case. Similarly, the LSEs can only drive their nodal prices up, and it can be shown that their payoff will not increase via unilateral deviation either. Thus, the bid profile given in (\ref{efficient_bid}) constitutes a Nash equilibrium under the monopoly-free condition as well. 
\end{proof}
The key idea behind both of these conditions is that there are other market participants in the network bidding such individual generators or LSEs can only cause prices to change in a direction which does not lead to increased payoff. In the congestion-free case, these other participants can be located anywhere in the network, while in the monopoly-free case, they are colocated at the same node. Note that while the monopoly-free condition can be checked from the network topology alone, the congestion-free condition requires knowledge of the optimal solution, given complete transparency on the part of all market participants. Even if such information is available, the specifics of a given problem instance, i.e.,  existence of congestion at the optimal dispatch, may still preclude application of the congestion-free condition for guaranteeing existence of an efficient Nash equilibrium. 

\section{Conclusions}\label{sec:conclusions}

In this paper, we have proposed a two-stage market mechanism that integrates renewable energy generators as an alternative to the extant multi-settlement markets that are operated independently though the decision-making of the market participants is obviously coupled. We formulate the two stage economic dispatch problem as a two-stage stochastic program with recourse. We first show that a sequential competitive equilibrium indeed exists in such a two-stage market. We show that every sequential competitive equilibrium supports an efficient allocation, and conversely every efficient allocation can be supported by a sequential competitive equilibrium.  We then design a market mechanism for such settings. We showed that when market participants act strategically, and if either a \textit{congestion-free} or a \textit{monopoly-free} condition is satisfied, the  Nash equilibrium  of the two-stage market mechanism  exists and is efficient. We also gave a counterexample that if LSEs bid marginal costs, then these conditions are not enough to guarantee of an efficient Nash equilibrium. We have ignored physical aspects of the network such as transmission loss, though that could be incorporated as well at the risk of greater notational and proof complexity though the essence of the results presented would be the same. Additional economic aspects such as demand elasticity are left to future work.
	
\bibliography{Bibliography}
\bibliographystyle{IEEEtran}

\appendices
\end{document}